\documentclass[11pt, a4paper]{article} 
\usepackage[margin=1in]{geometry}
\usepackage[utf8]{inputenc}
\usepackage[T1]{fontenc}
\usepackage{amsmath,amssymb,amsthm}
\usepackage[shortlabels]{enumitem}
\usepackage{color,url}
\usepackage{hyperref}
\usepackage[capitalise]{cleveref}
\usepackage{xspace}
\usepackage{todonotes}
\usepackage{thm-restate}
\usepackage{nicematrix}
\usepackage{lineno}
\usepackage[noadjust,compress]{cite}
\theoremstyle{plain}
\newtheorem{problem}{Open Problem}

\newtheorem{theorem}{Theorem}[section]
\newtheorem{lemma}[theorem]{Lemma}
\newtheorem{corollary}[theorem]{Corollary}

\newtheorem{definition}[theorem]{Definition}
\newtheorem{proposition}[theorem]{Proposition}
\newtheorem{observation}[theorem]{Observation}
\newtheorem{claim}{Claim}[theorem]

\newenvironment{claimproof}{\noindent\textit{Proof of Claim \theclaim:}}{\hfill$\lrcorner$} 

\title{Towards Tight Bounds for the Graph Homomorphism Problem Parameterized by Cutwidth via Asymptotic Rank Parameters}
\author{Carla Groenland, Isja Mannens, Jesper Nederlof, Marta Piecyk, Paweł Rzążewski}
\date{}

\newcommand{\core}{\mathsf{core}}
\renewcommand{\epsilon}{\varepsilon}
\DeclareMathOperator{\mim}{mim}
\DeclareMathOperator{\rank}{rank}
\DeclareMathOperator{\ctw}{ctw}
\DeclareMathOperator{\tw}{tw}
\DeclareMathOperator{\minrank}{minrank}
\newcommand{\mimsup}{\mathrm{mimsup}}
\newcommand{\hg}{\mathrm{him}}
\newcommand{\N}{\mathbb{N}}
\newcommand{\F}{\mathbb{F}}

\newcommand{\eps}{\epsilon}
\newcommand{\Oh}{\mathcal{O}}
\newcommand{\nand}{\mathrm{NAND}}

\newcommand{\homm}{\textsc{Hom}\xspace}
\newcommand{\homo}[1]{\textsc{Hom}(\ensuremath{#1})\xspace}
\newcommand{\lhomo}[1]{\textsc{LHom}(\ensuremath{#1})\xspace}

\newcommand{\cA}{\mathcal{A}}
\newcommand{\cB}{\mathcal{B}}

\newcommand{\cH}{\mathcal{H}}

\newcommand{\cR}{\mathcal{R}}

\newcommand{\tG}{\widetilde{G}}

\newcommand{\vphi}{\varphi}

\begin{document}
\maketitle

\begin{abstract}
A homomorphism from a graph $G$ to a graph $H$ is an edge-preserving mapping from $V(G)$ to $V(H)$.
In the graph homomorphism problem, denoted by $\textsc{Hom}(H)$, the graph $H$ is fixed 
and we need to determine if there exists a homomorphism from an instance graph $G$ to $H$.
We study the complexity of the problem parameterized by the \emph{cutwidth} of $G$, i.e., 
we assume that $G$ is given along with a linear ordering $v_1,\ldots,v_n$ of $V(G)$ such that, for each $i \in \{1,\ldots,n-1\}$, the number of edges with one endpoint in $\{v_1,\ldots,v_i\}$ and the other in $\{v_{i+1},\ldots,v_n\}$ is at most $k$.

We aim, for each $H$, for algorithms for \textsc{Hom}($H$) running in time $c_H^k n^{\mathcal{O}(1)}$ and matching lower bounds that exclude $c_H^{k \cdot o(1)}n^{\mathcal{O}(1)}$ or $c_H^{k(1-\Omega(1))}n^{\mathcal{O}(1)}$ time algorithms under the (Strong) Exponential Time Hypothesis. 
In the paper we introduce a new parameter that we call $\mathrm{mimsup}(H)$. 
Our main contribution is strong evidence of a close connection between  $c_H$ and $\mathrm{mimsup}(H)$:
\begin{itemize}
    \item an information-theoretic argument that the number of states needed in a natural dynamic programming algorithm is at most $\mathrm{mimsup}(H)^k$,
    \item lower bounds that show that for almost all graphs $H$ indeed we have  $c_H \geq  \mathrm{mimsup}(H)$, assuming the (Strong) Exponential-Time Hypothesis, and
    \item an algorithm with running time $\exp ( {\Oh( \mathrm{mimsup}(H) \cdot k \log k)}) n^{\mathcal{O}(1)}$.
    \end{itemize}
In the last result we do not actually assume that $H$ is a fixed graph. Thus, as a consequence, we obtain that the problem of deciding whether $G$ admits a homomorphism to $H$ is fixed-parameter tractable, when parameterized by cutwidth of $G$ and $\mathrm{mimsup}(H)$. 
    
The parameter $\mathrm{mimsup}(H)$ can be thought of as the $p$-th root of the maximum induced matching number in the graph obtained by multiplying $p$ copies of $H$ via certain graph product, where $p$ tends to infinity.
It can also be defined as an \emph{asymptotic rank parameter} of the adjacency matrix of $H$. 
Such parameters play a central role in, among others, algebraic complexity theory and additive combinatorics.
Our results tightly link the parameterized complexity of a problem to such an asymptotic rank parameter for the first time.
%
\end{abstract}

\thispagestyle{empty}
\clearpage
\setcounter{page}{1}

\section{Introduction}
The study of the fine-grained complexity of \textsf{NP}-hard problems parameterized by width parameters has recently received an explosive amount of attention.
In this study one aims to determine, for a given computational problem, a function $f$ such that (1) the problem can be solved in $f(k)n^{\Oh(1)}$ time on given instances formed by an $n$-vertex graph along with an appropriate decomposition with width $k$, and (2) any improvement to $f(k)^{o(1)}n^{\Oh(1)}$ or even $f(k)^{1-\Omega(1)}n^{\Oh(1)}$ time would violate a standard hypothesis, typically being respectively the Exponential Time Hypothesis (ETH) and the Strong Exponential Time Hypothesis (SETH).
Characterizing this complexity tightly often gives a deep insight in the combinatorial structure of the problem at hand, in particular about the relation that indicates when two `subsolutions' (for some definition of `subsolutions') combine into a global solution. An example where such insights had major consequences is \textsc{Hamiltonian Cycle} and the \textsc{Traveling Salesperson Problem}~\cite{DBLP:journals/iandc/BodlaenderCKN15,DBLP:journals/jacm/CyganKN18, DBLP:conf/stoc/Nederlof20}.

In contrast to the study of such fine-grained complexity, on the other side of the spectrum, a celebrated \emph{meta-theorem} by Courcelle~\cite{DBLP:journals/iandc/Courcelle90} shows that every graph property definable in the \emph{monadic second-order logic} can be decided in time $f(k) \cdot n$ on $n$-vertex graph given with a tree decomposition of width $k$.
While this is extremely general, it is not precise at all in the sense that the functions $f(k)$ given by Courcelle's theorem are typically doubly-exponential or more, while more tailored-algorithms with single-exponential functions exist. This begs the question: Could there be such a meta-theorem that gives a more fine-grained upper bound akin to the ones sought after above?
Unfortunately, such a fine-grained meta-theorem still seems out of reach, and many recent works apply some highly non-trivial problem-specific insights to actually get the combination of tight algorithms and lower bounds~\cite{DBLP:conf/esa/RooijBR09,DBLP:journals/iandc/BodlaenderCKN15,DBLP:journals/corr/KociumakaP17,DBLP:journals/eatcs/LokshtanovMS11,10.1007/978-3-642-22993-0_47,DBLP:conf/soda/CurticapeanLN18,
DBLP:conf/iwpec/BorradaileL16,DBLP:journals/dam/KatsikarelisLP19, DBLP:conf/icalp/MarxSS21, DBLP:conf/soda/CurticapeanM16}.


An intermediate step towards more general results such as Courcelle's theorem is to consider general problems that capture many natural well-studied problems as special cases. Such a step was already taken for certain locally checkable vertex subset problems, which capture natural problems including \textsc{Independent Set} and \textsc{Dominating Set}~\cite{DBLP:conf/soda/FockeMINSSW23}.
A particularly rich and elegant family of such problems can be defined via \emph{graph homomorphisms}.
A homomorphism from $G$ to $H$ is a mapping $\varphi: V(G) \rightarrow V(H)$ such that for every $uv \in E(G)$ we have $\varphi(u)\varphi(v) \in E(H)$. If $H$ is the complete graph on $k$ vertices, such mappings $\varphi$ are exactly proper $k$-colorings, and this is why these mappings are often referred to as a \emph{$H$-colorings} of $G$.
For a fixed graph $H$, by $\homo{H}$ we denote the computational problem in which one needs to determine whether there is a homomorphism to $H$ from an input graph $G$.
The complexity dichotomy for \homo{H} was provided by Hell and Ne\v{s}et\v{r}il~\cite{DBLP:journals/jct/HellN90}: \homo{H} is polynomial-time solvable if $H$ is bipartite, and \textsf{NP}-hard otherwise. So the cases relevant to our work are when $H$ is non-bipartite.

There has been impressive work on the complexity of \homo{H} in various settings~\cite{DBLP:conf/icalp/GanianHKOS22,DBLP:conf/stoc/CurticapeanDM17,DBLP:journals/algorithmica/ChitnisEM17,DBLP:journals/toct/ChauguleLV21,DBLP:conf/soda/RothW20,DBLP:conf/stacs/EgriKLT10,DBLP:conf/stoc/CaiM23,DBLP:conf/stoc/BulatovK22}. From the fine-grained perspective, a lot of attention was put in the parameterization by treewidth of the instance graph~\cite{DBLP:conf/stacs/EgriMR18,DBLP:conf/soda/FockeMR22,DBLP:conf/esa/OkrasaPR20,DBLP:journals/siamcomp/OkrasaR21,DBLP:journals/corr/abs-2210-10677}.
In particular, for \homo{H} and some its close relatives we exactly understand the fastest possible (up to the SETH) running time of algorithms parameterized by treewidth.
Perhaps even more interestingly, the techniques developed in this line of research led to a deep understanding of combinatorial properties of \homo{H} and its variants, and the results obtained on the way can be used far beyond the bounded-treewidth case.

Typically, the lower bounds for (a variant of) \homo{H} are shown by a reduction from (a variant of) $q$-\textsc{Coloring}, where the choice of $q$ depends on $H$. 
In particular,  Marx, Lokshtanov, and Saurabh~\cite{DBLP:journals/talg/LokshtanovMS18} showed that for any $q \geq 3$, the  $q$-\textsc{Coloring} problem on every instance $G$ cannot be solved in time  $(q - \epsilon)^{\tw(G)} \cdot |V(G)|^{\Oh(1)}$ for any $\epsilon >0$, unless the SETH fails (here $\tw(G)$ is the treewidth of $G$).
Similar lower bounds for $q$-\textsc{Coloring} are also known for other parameters, like \emph{cliquewidth}~\cite{DBLP:journals/siamdm/Lampis20}, \emph{feedback vertex set number}~\cite{DBLP:journals/talg/LokshtanovMS18}, \emph{vertex cover number}~\cite{DBLP:journals/dam/JaffkeJ23}, or \emph{component-order-connectivity}~\cite{DBLP:journals/corr/abs-2210-10677}.
A common element in all these results is that the constant in the base of the exponential factor in the complexity bound is an increasing function of the number $q$ of colors.

However, it appears that this is not the case for all natural width parameters.
For a linear ordering $v_1,v_2,\ldots,v_n$ of vertices of a graph $G$, its \emph{width} is the maximum number of edges between the sets $\{v_1,\ldots,v_i\}$ and $\{v_{i+1},\ldots,v_n\}$, over all $i \in \{1,\ldots,n-1\}$.
The \emph{cutwidth} of $G$, denoted by $\ctw(G)$, is the minimum width of a linear ordering of $V(G)$.

In stark contrast to the results listed above, Jansen and Nederlof~\cite{DBLP:journals/tcs/JansenN19} showed that for every $q$, the $q$-\textsc{Coloring} problem on instances $G$ given with a linear ordering of width $k$ can be solved in randomized\footnote{A slightly slower deterministic algorithm was also given.} time $2^{k} \cdot |V(G)|^{\Oh(1)}$. 
In particular, the base of the exponential factor does not depend on $q$.

This phenomenon appears to be very fragile, e.g., it no longer occurs for the \emph{counting} variant of $q$-\textsc{Coloring}~\cite{DBLP:conf/stacs/GroenlandMNS22}. In the context of the discussion above, it is very natural to ask about the situation for \homo{H}, i.e., whether the natural dynamic programming approach that works in time $|V(H)|^{k} \cdot |V(G)|^{\Oh(1)}$ can be improved. In particular, 
whether there exists an absolute constant $c$, such that for every graph $H$, the \homo{H} problem on $n$-vertex instances of cutwidth $k$ can be solved in time $c^k \cdot n^{\Oh(1)}$.
This question was answered in the negative by Piecyk and Rz\k{a}\.zewski~\cite{DBLP:conf/stacs/PiecykR21}, who showed that the base $c_H$ of the exponential factor in the complexity bound (seen as a function of $H$) grows to infinity even if we restrict ourselves to cycles.
More specifically, they show that $c_H$ is lower-bounded by the number of edges in a maximum \emph{induced matching}\footnote{An \emph{induced matching} of a graph is a set $M$ of edges such that the graph induced by the endpoints of $M$ is a matching.} in $H$, multiplied by 2.

Note that a maximum induced matching of a clique has only one edge, so this lower bound matches the running time of the randomized algorithm for $q$-\textsc{Coloring} by Jansen and Nederlof~\cite{DBLP:journals/tcs/JansenN19}. 

However, Piecyk and Rz\k{a}\.zewski~\cite{DBLP:conf/stacs/PiecykR21} failed to provide an algorithm matching this lower bound, even functionally, i.e., an algorithm with running time $f(p,k) \cdot n^{\Oh(1)}$, where $f$ is any function of the size $p$ of a maximum matching in $H$ and the cutwidth $k$ of the instance. Thus, while the size of a maximum induced matching in $H$ certainly plays some important role in the complexity of \homo{H} parameterized by the cutwidth, it is far from clear whether it  indeed determines the base of the exponential factor.
The discussion above leads to the following problem.
\begin{problem}\label{q1}
Describe, for any fixed non-bipartite graph $H$, a constant $c_H$ such that:
\begin{enumerate}
    \item There is an algorithm that, for all $k,n \in \mathbb{N}$, given an $n$-vertex graph $G$ with linear ordering of width $k$, solves  $\homo{H}$ in time $c_H^k \cdot n^{\Oh(1)}$, and
    \item Assuming the SETH, for any $\varepsilon >0$, 
    there is no algorithm that, for all $k,n \in \mathbb{N}$, given an $n$-vertex graph $G$ with linear ordering of width $k$, solves  $\homo{H}$ in time $(c_H-\varepsilon)^{k} \cdot n^{\Oh(1)}$.
\end{enumerate}
\end{problem}
Recall that when $H$ is bipartite, $\homo{H}$ is already known to be solvable in polynomial time.
Let us remark that for each graph $H$ we have $c_H \leq |V(H)|$, as a straightforward dynamic programming algorithm works in time $|V(H)|^k \cdot n^{\Oh(1)}$.

\subsection*{Our contribution.}
We make significant progress towards \cref{q1}.
In particular, for each non-bipartite graph $H$ we define a constant $c_H$ which we conjecture to have the desired properties.
We prove, for almost all graphs, that $\homo{H}$ in $n$-vertex instances given with a linear ordering of width $k$ cannot be solved in time $(c_H-\varepsilon)^{k} \cdot n^{\Oh(1)}$ for any $\varepsilon >0$, assuming the SETH. 
Moreover, we give a dynamic programming approach of which we show the table sizes can be compressed to $c_H^{k} \cdot n^{\Oh(1)}$ (see the paragraph on \emph{representative sets} below for more details).
This can be interpreted as an upper bound, for each $i \in \{1,\ldots,n-1\}$ on the amount of \emph{information} of the graph $G[\{v_1,\ldots,v_i\}]$ needed to decide $\homo{H}$ based on \emph{only} $G - G[\{v_1,\ldots,v_i\}]$.
Unfortunately, this is an existential result and we do not yet know how to efficiently perform this compression. We give partial progress towards such a computation, yielding an algorithm with running time $\exp(c_H)^{\Oh(k\log k)}n^{\Oh(1)}$. 

For a $0/1$-matrix $A$, we define $\mim(A)$ as the largest $r$ for which $A$ has an $r\times r$ permutation submatrix.\footnote{It is easily seen that $\mim(A)$ equals the maximum size of an induced matching in the bipartite graph that has $A$ as biadjacency matrix and, if $A$ is symmetric and hence the adjacency matrix of a graph $H$, twice the maximum size of an induced matching in $H$.}
The aforementioned work~\cite{DBLP:conf/stacs/PiecykR21} shows that $c_H$ needs to be at least $\mim(A_H)$, if $A_H$ is the adjacency matrix of $H$. 
However, one of our main insights is that, as the cutwidth $k$ increases, the accurate parameter for measuring the aforementioned amount of needed information on $G[\{v_1,\ldots,v_i\}]$ is actually $\mim(A_H^{\otimes k})$, where $A_H^{\otimes k}$ denotes the result of taking the Kronecker product of $k$ copies of $A_H$.
Specifically, we introduce a new asymptotic rank parameter, called \emph{mimsup}, defined by 
\[
\mimsup(A)=\limsup_{k\to\infty }\mim(A^{\otimes k})^{1/k}.
\]
For a graph $H$, we define $\mimsup(H)$ to be  $\mimsup(A_H)$, where $A_H$ is the adjacency matrix of $H$.  We remark that $\mimsup(H)$ can be also defined in purely graph-theoretic way, in terms of the size of a maximum matching in a certain graph product. See Section~\ref{sec:prel} for more thorough definitions and details.
We prove the results above for $c_H$ equal to  $\mimsup(H)$ (modulo some standard preprocessing of $H$). 


The parameter becomes especially clean and elegant if $H$ is a \emph{projective core}; such graphs play a prominent role in the study of graph homomorphisms~\cite{larose2001strongly,DBLP:journals/siamcomp/OkrasaR21,DBLP:conf/icalp/GanianHKOS22}.
Their definition is somewhat complicated, so we postpone it to \cref{sec:algofinal}.
Let us just mention that this class captures \emph{almost all graphs}~\cite{DBLP:journals/dm/92,luczak2004note}.
Formally, let $p_n$ denote the probability that an $n$-vertex graph, chosen uniformly at random, is a non-bipartite projective core. Then $p_n$ tends to 1 as $n$ tends to infinity.
Furthermore, up to some conjectures from algebraic graph theory from the early 2000s~\cite{larose2001strongly,Larose2002FamiliesOS}, every graph $H$ that cannot be simplified by the above-mentioned preprocessing is actually a projective core. We refer the interested reader to~\cite{DBLP:journals/siamcomp/OkrasaR21} for more information.

Going back to our setting, if $H$ is a non-bipartite projective core, then we simply have $c_H = \mimsup(H)$.
The first evidence that $c_H$ is indeed the ``right'' choice of the parameter is the following lower bound.

\begin{restatable}{theorem}{thmlower}
\label{thm:lower}    
\begin{enumerate}
\item There exists $\delta>0$, such that for every non-bipartite projective core $H$, there is no algorithm solving every instance $G$ of $\homo{H}$ in time $\mimsup(H)^{\delta\cdot\ctw(G)}\cdot n^{\Oh(1)}$, unless the ETH fails.
\item Let $H$ be a connected non-bipartite projective core. There is no algorithm solving every instance $G$ of $\homo{H}$ in time $(\mimsup(H)-\eps)^{\ctw(G)}\cdot |V(G)|^{\Oh(1)}$ for any $\eps>0$, unless the SETH fails.
\end{enumerate}
\end{restatable}

We next elaborate on the mentioned dynamic programming approach along with the table size compression via which we aim to match these lower bounds.

\paragraph{Representative Sets.}
A crucial technique in our arguments is that of \emph{representative sets}. This is a method that allows us to considerably speed up dynamic programming algorithms by sparsifying the associated tables. 
Specifically, dynamic programming algorithms generally define a space of possible partial solutions $\mathcal{S}$, and a dynamic programming table stores a subset $\mathcal{A}$ of partial solutions that are valid in the given instance. A binary \emph{compatibility matrix} $M$ with rows and columns indexed by $\mathcal{S}$ indicates whether two partial solutions combine into a global solution.
Generally speaking, a \emph{representative set} of a set $\mathcal{A} \subseteq \mathcal{S}$ is a subset $\mathcal{A}'$ such that for each $j \in \mathcal{S}$ we have that there exists $i \in \mathcal{A}$ with $M[i,j]=1$ if and only if there exists $i'\in \mathcal{A}'$ with $M[i',j]=1$; see Section~\ref{sec:def_repsets} for details more specific to the setting of our paper.

The power of representative sets lies in that (i) by definition, in any dynamic programming algorithm we can replace the set $\mathcal{A}$ with the smaller set $\mathcal{A}'$ without missing solutions, and (ii) for many matrices $M$, surprisingly small representative sets are guaranteed to exist. This underlies, for example, fast algorithms for the $k$-\textsc{Path} problem~\cite{DBLP:conf/wg/Monien83} or connectivity problems parameterized by treewidth~\cite{DBLP:journals/iandc/BodlaenderCKN15, DBLP:journals/jacm/FominLPS16}.
However, a serious bottleneck in these algorithm is the \emph{computation} of such representative sets: It withholds us, for example, for getting faster algorithms for connectivity problems such as \textsc{Traveling Salesperson}  (both parameterized by treewidth~ \cite{DBLP:journals/iandc/BodlaenderCKN15, DBLP:journals/jacm/CyganKN18} and the classic parameterization by the number of cities~\cite[Theorem 3]{DBLP:conf/stoc/Nederlof20}), and polynomial kernelization algorithms for~\textsc{Odd Cycle Transversal}~\cite{DBLP:journals/jacm/KratschW20}.

This already led some researchers to design faster algorithms for finding representative sets in special settings. A natural setting that comes up, for example for connectivity problems parameterized by treewidth, is to find representative sets for sets of partial solution with a certain \emph{product structure}. In~\cite{DBLP:journals/talg/FominLPS17}, the authors show that representative sets for such families can be found faster than known for general families.

In this paper, the computation of representative sets is also a major bottleneck; in fact, modulo the standard conjectures discussed above, it is the only issue that withholds us from solving Open Problem~\ref{q1} completely. Specifically, we show:

\begin{theorem}[Informal statement of Theorem~\ref{thm:reduce}]
In the context of the natural dynamic programming algorithm for \homo{H} parameterized by cutwidth $k$, there exist representative sets of size at most $\mimsup(H)^k$.
\end{theorem}

Thus, by the definition of representative sets, any algorithm that computes these representative sets fast enough would imply a $\mimsup(H)^{\ctw(G)}n^{\Oh(1)}$ time algorithm for $\homo{H}$ and thus solve Open Problem~\ref{q1}. 
We view this as strong evidence that our lower bounds cannot be improved.
Indeed, state-of-the-art hardness reduction techniques (like~\cite{DBLP:journals/talg/LokshtanovMS18}) for problems parameterized by width parameters encode assignments to decision variables as states of dynamic programming tables and gradually check constraints on and global consistency of these assignments throughout the graph.
Our proof of existence of small representative sets means the number of assignments that need to be considered in order to find a global solution is also small, which means that this kind of approach to design lower bounds hits a natural barrier at our lower bound.

\paragraph{Coping Algorithmically with the Mimsup Parameter.}
Matrix or graph parameters that are defined in terms of large powers are sometimes called \emph{asymptotic rank parameters}, and they are notoriously hard to compute. For example, the value of  the asymptotic rank of the matrix multiplication tensor~\cite{Brgisser1997AlgebraicCT} (also known as $\omega$) or the Shannon capacity of the cycle on $7$ vertices~\cite{godsil} remain elusive.
Unfortunately, $\mimsup(H)$ seems no exception. Similarly to the Shannon capacity~\cite{AlonBallaGishbolinerMondMousset20}, it is even not clear whether computing it is decidable.
For $\mimsup(H)$, even for simple graphs such as $H=K_q$, determining its value is non-trivial as well. As an illustration we depict the maximum induced matching in the second Kronecker power of the adjacency matrix of $K_3$ in Figure~\ref{fig:color}.
One of the main insights of the $2^{\ctw(G)}n^{\Oh(1)}$ time algorithm for the chromatic number from~\cite{DBLP:journals/tcs/JansenN19} shows that $\mimsup(K_q)$ in fact equals $2$; we revisit their proof in our new language in Section~\ref{sec:suprank}.

\begin{figure}
\[
\begin{NiceArray}{|ccc|ccc|ccc|}
\hline
\textcolor{red}{\textbf{0}}& \textcolor{red}{\textbf{0}}& 0 & \textcolor{red}{\textbf{0}}&\textcolor{red}{\textbf{1}} &1  &0 & 1& 1\\
\textcolor{red}{\textbf{0}}& \textcolor{red}{\textbf{0}}& 0& \textcolor{red}{\textbf{1}} & \textcolor{red}{\textbf{0}} & 1 & 1 &  0 & 1 \\
0&0 & 0& 1 & 1 & 0 & 1 &  1 & 0 \\
\hline
\textcolor{red}{\textbf{0}} & \textcolor{red}{\textbf{1}} & 1 & 
\textcolor{red}{\textbf{0}} & \textcolor{red}{\textbf{0}}& 0
&0 & 1& 1\\
\textcolor{red}{\textbf{1}} & \textcolor{red}{\textbf{0}} & 1 & \textcolor{red}{\textbf{0}}& \textcolor{red}{\textbf{0}}& 0& 1 &  0 & 1 \\
1 & 1 & 0 & 0& 0& 0& 1 &  1 & 0 \\
\hline
0 & 1 & 1 &0 & 1& 1& 0&0&0\\
1 & 0 & 1 & 1 &  0 & 1 & 0&0 &0 \\
1 & 1 & 0 & 1 &  1 & 0 & 0& 0& 0\\
\hline
\end{NiceArray}
\]
\caption{Illustration of a maximum induced matching (or equivalently, induced permutation submatrix) of size $\mim(A_{H}^{\otimes 2})$ shown in red, where $A_{H}=\Bigl( \begin{smallmatrix}0 & 1 & 1 \\ 1 & 0 & 1 \\ 1 & 1 & 0\end{smallmatrix}\Bigr)$ and $H=K_3$. More generally, the proof from~\cite{DBLP:journals/tcs/JansenN19} (further discussed in Section~\ref{sec:suprank}) for determining the chromatic number of a graph shows that whenever $H$ is a complete graph, $\mim(A_{H}^{\otimes k})=2^k$ and thus $\mimsup(H)=2$. }
\label{fig:color}
\end{figure}
Even when the existence of small representative sets is guaranteed because $\mimsup(H)$ is small, it is still challenging to find them quickly.
Since $\mimsup$ is about products of graphs, one may expect that this product structure can be used algorithmically. Indeed, product structure has been exploited to compute representative sets in previous work~\cite{DBLP:journals/talg/FominLPS17}, but there the family that needs to be represented has some (Cartesian) product structure. In our setting, this is not guaranteed and it is much less obvious how to proceed.

Nevertheless we show that, with some loss in precision (and hence, running time), we can work with graphs with small $\mimsup$, by approximating it with another (easier to compute) value that we call the maximum \emph{half induced matching} number $\hg(H)$. For a matrix $A$, we define $\hg(A)$ as the largest $r$ for which $A$ has an $r\times r$ triangular submatrix with ones on its diagonal.\footnote{A submatrix of a matrix $A$ is any matrix that can be obtained from $A$ by removing and reordering any of its rows and columns.} 
For a graph $H$, we define $\hg(H)$ to be  $\hg(A_H)$, where $A_H$ is the adjacency matrix of $H$. We show that $\hg(A_H)$ approximates $\mimsup(A_H)$ in the following sense: 
\begin{restatable}{theorem}{combhgbnd}
\label{thm:combhgbound}
 For every non-bipartite graph $H$ with adjacency matrix $A_H$ and $k\in \N$, 
 \[
 \hg(A_H)\leq \mimsup(A_H)=\lim_{k\to \infty} \mim(A_H^{\otimes k})^{1/k}\quad\text{and}\quad\mim(A_H^{\otimes k})\leq k^{\hg(A_H)k}.
 \]
\end{restatable}
The parameter $\hg(A_H)$ is easily computable in time $2^{\Oh(|V(H)|)}$. 
While the lower bound on $\mimsup(A_H^{\otimes k})$ is relatively easy, the upper bound uses an argument similar to the `neighborhood chasing' argument for the upper bounds on multi-colored Ramsey numbers~\cite{erdos1935combinatorial}. This argument can in fact be made algorithmic in the sense that it can be used to compute representative sets for $\homo{H}$ of size at most $\Oh(k^{\hg(H)k})$ in time $\Oh(k^{2\hg(H)k})$.
Combining this result with the dynamic programming algorithm for \homo{H} for graphs of small cutwidth yields the following result.

\begin{sloppypar}
\begin{restatable}{theorem}{thmMimAlg}
\label{thm:MimAlg}
   For any graphs $G$ and $H$, where $G$ given with a linear ordering of width $k$, in time $\Oh(k^{2k\cdot\mimsup(H)} \cdot |V(H)|^4|V(G)|)$ one can decide whether $G$ admits a homomorphism to $H$.
\end{restatable}
\end{sloppypar}

Let us compare the running time in Theorem~\ref{thm:MimAlg} with the naive approach; recall that its complexity is $|V(H)|^k \cdot |V(G)|^{\Oh(1)}$.
If we treat $H$ as a constant and $k$ as a parameter, then the latter one is faster. However,  we emphasize here that in Theorem~\ref{thm:MimAlg} we do not assume that $H$ is a constant, so these two algorithms are incomparable.
In particular, Theorem~\ref{thm:MimAlg} shows that the homomorphism problem, where the input consists of both $G$ and $H$, is \emph{fixed-parameter-tractable} when parameterized by the cutwidth of $G$ and $\mimsup(H)$.

It should be noted that a similar notion of half-induced matching of a compatibility matrix was already introduced in previous work in the context of representative sets of the \textsc{AntiFactor} problem~\cite{marx_et_al:LIPIcs.IPEC.2022.22} parameterized by treewidth and list size.
However, in that setting, the authors were only able to provide a lower bound for their problem, and they did not manage to make the connection with half-induced matchings algorithmic. Additionally, their compatibility matrix has a very specific structure:  it is indexed with integers and the value of an entry only depends on the sum of the values associated with the row and column.

\paragraph{Comparison of Mimsup With Other Rank Parameters}
One of our main conceptual contributions is the introduction of the $\mimsup$ parameter in the context of parameterized algorithms. It is actually the first asymptotic rank parameter shown to be relevant in this context.
Our $\mimsup$ parameter is very similar to the \textit{asymptotic induced matching number} studied by Arunachalam et al.~\cite{asymptoticinducedmatching} which was introduced for $k$-partite, $k$-uniform hypergraphs (and so, in the graph setting, only  for bipartite graphs).
Various asymptotic variants of rank parameters have been studied for tensors. For example, this has been done for rank parameters such as subrank, tensor rank and slice rank. However, for matrices (2-tensors) these are equal to the `standard' rank for matrices and so have no interesting asymptotic aspects.


That being said, it is only natural to compare the $\mimsup$ parameter with different related rank parameters from the literature. We will discuss this now, and provide proofs that formally support this discussion in Section~\ref{sec:cis}.
%
The approach from~\cite{DBLP:journals/tcs/JansenN19} naturally extends to solve $\homo{H}$ quickly for all graphs $H$ where the so-called \emph{support-rank} of the adjacency matrix of $H$ is small. We provide details of this in Section ~\ref{sec:suprank}. We will show the following sequence of inequalities for the parameters discussed so far:
\[
    \mim(A)\leq \hg(A)\leq \mimsup(A)\leq \text{support-rank}(A)\leq \rank(A).
\]
The first inequality is direct and we will give a short proof of the second inequality in \cref{lem:mimsupgeqhg}. The third inequality is detailed in Observation~\ref{obs:supprankvsmimsup}, and the fourth inequality follows directly from the definition of support rank.

We believe all of the inequalities can be strict. When $A$ is the $(r\times r)$-matrix with ones on and above the diagonal and zeros below the diagonal, then $\mim(A)=1<r=\hg(A)$ for $r\geq 2$. This means that $\mim$ is not functionally equivalent to $\hg$ nor $\mimsup$. We use random matrices to show that $\hg$ and $\mimsup$ may take very different values in \cref{thm:hgversusmimsup}. We also show that $\hg$ is not functionally equivalent to the support rank in \cref{thm:separation_hg_minrank}. This shows our algorithm from Theorem~\ref{thm:MimAlg} can be significantly faster than the discussed natural generalization of~\cite{DBLP:journals/tcs/JansenN19}. We leave it as an interesting open problem whether mimsup is functionally equivalent to him or support-rank.


The aforementioned (well studied) Shannon capacity has a definition that is very similar to the $\mimsup$ parameter: It is defined in terms of the maximum size of an independent set (also called the independence number) in an appropriate graph product, and the size of a maximum induced matching of a graph equals the independence number of the square of its line graph. Unfortunately, because the definitions of $\mimsup$ and Shannon capacity use different graph products, the relation between the two is somewhat loose; see Section~\ref{sec:shannon} for details. Nevertheless, based on their similarity, one may expect that Shannon capacity shares some of its peculiarities with $\mimsup$, such as an unpredictable behaviour of the value in graph powers~\cite{DBLP:journals/tit/AlonL06}.

\paragraph{Organization.}
This paper is organized as follows.
In Section~\ref{sec:prel} we provide some preliminary notation and tools used in this paper.
In Section~\ref{sec:whyrepsets} we provide the mentioned dynamic programming that solves $\homo{H}$ with representative sets.
In Section~\ref{sec:algo} we present our algorithms to compute representative sets.
In Section~\ref{sec:algofinal} we prove Theorem~\ref{thm:mainalgohomo} and Corollary~\ref{cor:algohomoprojective}.
In Section~\ref{sec:lower} we prove the lower bounds from Theorem~\ref{thm:lower}.
In Section~\ref{sec:cis} we prove our combinatorial results, amongst others showing him and $\mimsup$ are different parameters.
In Section~\ref{sec:conc} we provide concluding remarks about cases that are potentially not covered by our results, and some directions for further study. 

\section{Preliminaries}
\label{sec:prel}
For an integer $n$, by $[n]$ we denote the set $\{1,2,\ldots,n\}$ and for integers $a,b$ we write $[a,b]=\{a,a+1,\dots,b\}$.
For a set $X$, by $2^X$ we denote the family of all subsets of $X$.
For a graph $G$ and $V' \subseteq V(G)$ (resp. $E' \subseteq E(G)$), by $G[V']$ (resp. $G[E']$) we denote the subgraph of $G$ induced by $V'$ (resp. $E'$).

\paragraph{Homomorphisms.}
For graphs $G$ and $H$, a \emph{homomorphism} from $G$ to $H$ is a mapping $\vphi : V(G) \to V(H)$ such that for every $uv \in E(G)$ we have $\vphi(u)\vphi(v) \in E(H)$.
If $\vphi$ is a homomorphism from $G$ to $H$, we denote it by writing $\vphi : G \to H$. 
If $G$ admits a homomorphism to $H$, we denote is shortly by $G \to H$.

In the \homm problem we are given a pair $(G,H)$ of graphs, and we ask whether $G \to H$.
In the \homo{H} the graph  $H$ is considered to be fixed and we ask whether a graph $G$ given as an input admits a homomorpshism to $H$.

We will always assume that $H$ is a connected graph.
Indeed, each component of $G$ must map to a single component of $H$, so the problem can be solved component-wise.

\paragraph{Cutwidth.}
Let $G$ be a graph and consider a linear ordering $\sigma = (v_1,\ldots,v_n)$ of its vertices.
For $i \in [n-1]$, the \emph{$i$-th cut} is the partition of $V(G)$ into sets $\{v_1,\ldots,v_i\}$ and $\{v_{i+1},\ldots,v_n\}$.
The \emph{width} of such a cut is the number of edges with one endvertex in $\{v_1,\ldots,v_i\}$ and the other in $\{v_{i+1},\ldots,v_n\}$. 
The \emph{width} of $\sigma$ is the maximum width of a cut of $\sigma$.
Finally, the \emph{cutwidth} of $G$, denoted by $\ctw(G)$, is the minimum width of a linear ordering of the vertices of $G$.

\paragraph{Associated bipartite graphs.}  In order to define the main parameters of our paper, we will use a notion of \emph{associated bipartite graphs}. For a graph $G$, the graph $G^*$ is defined as follows.
\begin{align*}
    V(G^*) & = \{u', u'' \ | \ u\in V(G)\}, \\
    E(G^*) & = \{u'w'', u''w' \ | \ uw\in E(G)\}.
\end{align*}

\paragraph{Induced matchings and half-induced matchings.}
A set $M\subseteq E$ of edges of a graph $H=(V,E)$ forms an \textit{induced matching} if the edges in $M$ are vertex disjoint and no edge in $E$ is incident with two edges from $M$. We may also view this as two sequences of distinct vertices $v_1,\dots,v_m$ and $u_1,\dots,u_m$ where $v_iu_j\in E$ if and only if $i=j$. For a bipartite graph $H$, by $\mim(H)$ we denote the size of a maximum induced matching in $H$. For non-bipartite $H$, we define $\mim(H):=\mim(H^*)$.


A \textit{half-induced matching} of a graph $H$ consists of two sequences $v_1,\dots,v_m$ and $u_1,\dots,u_m$ of distinct vertices  where $v_iu_i\in E$ for $i\in [m]$ and $u_iv_j\not\in E$ if $1\leq i<j\leq m$. For a bipartite graph $H$, we denote the size of the largest half-induced matching in $H$ by $\hg(H)$. 
We extend the definition to graphs $H$ that are non-bipartite via $\hg(H)=\hg(H^*)$.
This notion has been studied under the name  \textit{constrained matching} (a subset with a unique matching, see e.g. \cite{Chapman,OLESKY1993183,trefois2015zero}), but we decided to use the name which appeared more recently in a similar setting to ours  \cite{marx_et_al:LIPIcs.IPEC.2022.22}, since the word `constrained matching' has also been used for various other purposes in the algorithmic community.

\paragraph{Mim and him for matrices.}
Let $A\in \{0,1\}^{n\times n}$ be a matrix. Given a sequence $r\in [n]^a$ of distinct row indices and $c\in [n]^b$ of distinct columns indices, for some integers $a,b\in [n]$, we write $A[r,c]$ for the $a\times b$ matrix with entries $A[r,c]_{i,j}=A_{r_i,c_j}$ for $i\in [a]$ and $j\in [b]$. We refer to any matrix which arises in such a manner as a \emph{submatrix after permutation} of $A$.

We write $\mim(A)$ for the maximum $r$ for which $A$ has the $r\times r$ identity matrix as submatrix after permutation (equivalently, the largest permutation submatrix). 
We write $\hg(A)$ for the largest $r$ for which $A$ has an $r\times r$ triangular matrix with ones on the diagonal as submatrix after permutation. We will also refer to such a submatrix as \emph{half induced matching}. (A matrix is called \textit{triangular} if either all entries below the diagonal, or all entries above the diagonal are 0.)

For bipartite graphs $H=(U,V,E)$, the \textit{bi-adjacency matrix} $B$ is indexed by rows from $U$ and columns from $V$ where $B[u,v]=1$ if $uv\in E$ and $B[u,v]=0$ otherwise. For a bipartite graph $H$, there is a one-to-one correspondence between induced matchings of $H$ of size $m$ and  $m\times m$ identity submatrices of the bi-adjancency matrix of $H$. In particular, for a bi-adjacency matrix $B$ of $H$, $\mim(B)=\mim(H)$. Similarly, $\hg(B)=\hg(H)$.

For a non-bipartite graph $G$, if $A_G$ is its adjacency matrix, then $A_G$ is also the bi-adjacency matrix of~$G^*$. This means that for non-bipartite $H$ with adjacency matrix $A_H$, $\mim(H)=\mim(A_H)$ and $\hg(H)=\hg(A_H)$.

\paragraph{Mimsup.}
For a matrix $A$, we define
\[
\mimsup(A)=\limsup_{k\to \infty}\mim(A^{\otimes k})^{1/k}.
\]
Here $\otimes$ denotes the Kronecker product of the matrix. Given an $n\times m$ matrix $A=(a_{i,j})_{i\in [n],j\in [m]}$ and a matrix $B$, the Kronecker product is given by
\[
A\otimes B = \begin{pmatrix} a_{1,1} B& a_{1,2} B&\dots & a_{1,m} B\\
a_{2,1} B& a_{2,2} B&\dots & a_{2,m}B\\
&\dots & \\
a_{n,1} B& a_{n,2} B& \dots & a_{n,m} B\\
\end{pmatrix}.
\]
Since $\mim(A\otimes B)\geq \mim(A)\mim(B)$, Fekete's lemma~\cite{Fekete1923} applies to show that 
\[
\limsup_{k\to \infty}\mim(A^{\otimes k})^{1/k}=\lim_{k\to \infty}\mim(A^{\otimes k})^{1/k}=\sup_{k\in \N}\mim(A^{\otimes k})^{1/k}.
\]
Indeed, Fekete's lemma states that for any  sequence $(a_n)_{n\in \mathbb{N}}$ which is subadditive ($a_{n+m}\leq a_n+a_m$ for all $n,m
 \in \mathbb{N}$), the limit $\lim_{n\to \infty}a_n/n$ exists (and equals $\inf_n a_n/n$). Since mim is supermultiplicative, the sequence defined by \[
 a_k=-k\log(\mim(A^{\otimes k})^{1/k})=-\log(\mim(A^{\otimes k}))
 \]
 is subadditive. This means the limit \[
 \lim_{k\to \infty }\frac1k a_k= \lim_{k\to \infty }-\log(\mim(A^{\otimes k})^{1/k})
 \]
 exists. That, in turn, implies the limit $ \lim_{k\to \infty }\mim(A^{\otimes k})^{1/k}$ exists.

For a non-bipartite graph $H$, with adjacency matrix $A$, we set
\[
\mimsup(H)= \mimsup(A).
\]
When $H$ is bipartite with bi-adjancency matrix\footnote{Note that mimsup is invariant under row and column permutations. This means that the choice of bi-adjacency matrix does not affect the mimsup and thus mimsup on bipartite graphs is well-defined.} $B$, $\mimsup(H)=\mimsup(B)$. The parameters can also be defined in purely graph theoretical terms, as we now explain.

For a bipartite graph $H$ with bipartition classes $X,Y$, and for $k\in\N$, we define $H^{\otimes k}$ to be the graph on vertex set $X^k\cup Y^k$ where there is an edge $(x_1,\ldots,x_k)(y_1,\ldots,y_k)$ in $H^{\otimes k}$ if and only if $x_iy_i\in E(H)$ for every $i\in[k]$. 
With this definition of graph product $\otimes$, we define
\[
 \mimsup(H)=
\begin{cases}
\limsup_{k\to \infty}\mim(H^{\otimes k})^{1/k} & \text{ if $H$ is bipartite,}\\
\mimsup(H^*)                            & \text{ otherwise.}
\end{cases} 
\]


The following property of $\mimsup$ is straightforward.
\begin{observation}\label{obs:mimsup-subgraph}
    If $H$ is an induced subgraph of $G$, then $\mimsup(H)\leq \mimsup(G)$.
\end{observation}
For bipartite graphs, $\mimsup$ coincides with the parameter asymptotic induced matching number studied by \cite{asymptoticinducedmatching}. Although asymptotic rank parameters (e.g. asymptotic subrank, asymptotic tensor rank and asymptotic slice rank) have been widely studied for tensors, the `non-asymptotic' parameters are usually equal to the matrix rank for matrices, which has no interesting asymptotic behaviour since $\rank(A^{\otimes n})=\rank(A)^n$. In particular, the \textit{subrank} in some sense looks for the largest `identity subtensor', similar to our $\mim$, but since it allows row operations to be applied (instead of merely permutations), this notion is the same as the usual rank for matrices and the same holds for the asymptotic subrank.

\section{Solving \homm with representative sets}
\label{sec:whyrepsets}
In this section we discuss how we can use representative sets to create fast algorithms for \homm.
We start by giving a definition of a representative set in our setting.
Intuitively we want a representative set $\mathcal{A}'$  of $\mathcal{A}$ to carry all the important information from $\mathcal{A}$, while being smaller in size. 
In practice, being able to find small representative sets corresponds to having to compute less entries in a dynamic programming algorithm. 
So this gives the following natural extremal problem: how small of a representative are we always guaranteed to find, that is, \emph{what is the largest size of a set which has no smaller representative set?} 
After giving the definition, we explain why $\mimsup$ exactly determines the answer to this question in our setting.

Finally, we give a general framework for solving \homm instances; it  consists of an algorithm that takes as input some reduction algorithm $R$ that produces small representative sets and uses it to solve \homm on input graphs $G$ and $H$. In Section~\ref{sec:algo} we give examples of such reduction algorithms.

\subsection{Definition of Representative Set}
\label{sec:def_repsets}
Given a $0/1$ matrix $M$, with rows indexed by a set $\mathcal{R}$ and $\mathcal{A}\subseteq \mathcal{R}$, we are interested in knowing whether for a column $c$, there is a row $r\in \mathcal{A}$ with $M[r,c]=1$. In our case, 
\begin{itemize}
    \item each row represents a coloring of the left-hand side of the cut;
    \item $\mathcal{A}$ contains the colorings that can be extended to the left-hand side of the (input) graph;
    \item each column represents a coloring of the right-hand side of the cut;
    \item $M[r,c]=1$ if and only if the colorings represented by row $r$ and column $c$ are compatible. 
\end{itemize}
This makes the following definition very natural.
We say that a subset $\mathcal{A}' \subseteq \mathcal{A}$ \emph{$M$-represents} $\mathcal{A}$, if for any column $j$ we have that if there is a row index $i \in \mathcal{A}$ such that $M[i,j] = 1$, then there is also $i' \in \mathcal{A}'$ such that $M[i',j] = 1$. Intuitively, this means that we do no `lose any solutions' by restricting to $\mathcal{A}'$.

We will also refer to $\mathcal{A}'$ as an \emph{$M$-representative set} of $\mathcal{A}$. We may omit $M$ if it is clear from context.

We remark that representing  is transitive: if $\mathcal{A''}$ represents $\mathcal{A}'$ and $\mathcal{A}'$ represents $\mathcal{A}$, then $\mathcal{A''}$ represents $\mathcal{A}$.

Suppose we aim to solve $\homm$ for input graphs $G$ and $H$, where $H$ is non-bipartite.
We will be interested in representative sets with respect to $M=A_H^{\otimes k}$ for integers $k$, where $A_H$ is the adjacency matrix of $H$.
We assume that $G$ is given with a linear order $v_1,\dots,v_n$ of width at most $w$.
For an integer $i\in [n]$ we refer to $G[\{v_1,\dots,v_i\}]$ as the `left-hand side of the $i$-th cut'
and \[
X_i := \left\{v \in \{v_1, \dots, v_i\} \ | \ \exists v' \in \{v_{i+1}, \dots, v_n\}, vv' \in E(G) \right\}.
\]
as `the left-hand side of the $i$-th cut.' 
Suppose there are $k$ edges crossing the cut: $\{a_1,b_1\},\dots,\{a_k,b_k\}\in E(G)$ with $a_1,\dots,a_k\in \{v_1,\dots, v_i\}$ and $b_1,\dots,b_k\in \{v_{i+1},\dots, v_n\}$. Let $L_i=(a_1,\dots,a_k)$ and $R_i=(b_1,\dots,b_k)$. Note that $\{a_1,\dots,a_k\}=X_i$ but some elements may be repeated.
A row $r$ of the matrix $M=A_H^{\otimes k}$ is a $k$-tuple $(r_1,\dots,r_k)\in V(H)^k$, which corresponds to a coloring $X_i\to V(H)$ if $r_j=r_{j'}$ whenever $a_j=a_{j'}$.
If similarly $c\in V(H)^k$ represents a coloring of the `right-hand side of the cut', then $M[r,c]=1$ if and only if $r_jc_j\in E(H)$ for all $j\in [k]$, i.e. the colorings are compatible. So indeed we capture the properties informally claimed above.

Since in our setting $M$ will be the adjacency matrix of some graph $H$, we may refer to $H$-representative sets rather than $A_H$-representative sets. 

\subsection{Connection to Mimsup}
We now show that the largest size of a set $\mathcal{A}\subseteq V(H)^k$ without smaller $A_H^{\otimes k}$-representative set, equals $\mimsup(H)^k$ (for $k$ an integer, $H$ a non-bipartite graph and $A_H$ its adjacency matrix). This easily follows from the definitions but is still one of the key conceptual contributions of this paper.
\begin{theorem}
\label{thm:reduce}
Let $H$ be a non-bipartite graph on $h$ vertices and let $A_H$ be its adjacency matrix. 
\begin{itemize}
    \item For each integer $k\in \mathbb{N}$, for any $\mathcal{A} \subseteq V(H)^k$, there is a subset $\mathcal{A}'\subseteq \mathcal{A}$ of size $\mimsup(H)^k$ that $A_H^{\otimes k}$-represents $\mathcal{A}$.
    \item Conversely, for each $\epsilon>0$, for each sufficiently large $k$, there is a $\mathcal{A} \subseteq V(H)^k$, for which no $\mathcal{A}'\subseteq \mathcal{A}$ of size at most $(\mimsup(H)-\epsilon)^k$ can $A_H^{\otimes k}$-represent $\mathcal{A}$.
\end{itemize}
\end{theorem}
\begin{proof}
Let $M=A_H^{\otimes k}$. 
Let $\mathcal{A}'\subseteq \mathcal{A}$ be of minimal size among the subsets that $M$-represent $\mathcal{A}$. Then no proper subset of it $M$-represents $\mathcal{A}$.
This means that for each $a \in \mathcal{A}'$ it cannot be removed from $\mathcal{A}'$ to get a set that $M$-represents $\mathcal{A}$.
Thus, for each $a \in \mathcal{A}'$ there is some $\mu(a) \in V(H)^k$ such that $M[a,\mu(a)]=1$, but for every $a' \in \mathcal{A}'\setminus \{a\}$ we have that $M[a',\mu(a)]=0$.
Hence $\{ a\mu(a) : a \in \mathcal{A}'\}$ is an induced matching in $H^{\otimes k}$ of size $|\mathcal{A}'|$. This shows that $|\mathcal{A}'| \leq\mim(A_H^{\otimes k})$. By definition of $\mimsup$, $\mim(A_H^{\otimes k})\leq  \mimsup(H)^k$.

Conversely, by definition of limit, for each $\epsilon>0$ there is a $k_0$ such that $\mim(A_H^{\otimes k})\geq (\mimsup(A_H)-\epsilon)^k$ for all $k\geq k_0$. 
Any induced matching has no smaller representative sets, so it suffices to consider the `left-hand side' $\mathcal{A}\subseteq V(H)^k$ of an induced matching in $H^{\otimes k}$ of size $\mim(A_H^{\otimes k})=\mim(H^{\otimes k})$.
\end{proof}


\subsection{Exploiting Representative Sets in Dynamic Programming}
The main idea behind the use of representative sets in an algorithmic setting is as follows. We solve the problem with a standard dynamic programming approach, where the cells are indexed by the elements of the set $\mathcal{A}$. A representative set then forms a small subset of these indices, which still carries enough information to solve the problem. By regularly applying the reduction algorithm, we can effectively run our dynamic programming algorithm on only a small subset of the cells in the table. We formalize this in the following theorem.
Let us emphasize that $H$ is not assumed to be fixed here but rather given as an input. 

\begin{theorem} \label{thm:MetaAlg}
    Let $H$ be a non-bipartite graph on $h$ vertices. Let $R$ be a reduction algorithm that, given an integer $k\geq 2$ and a subset $\mathcal{A} \subseteq V(H)^k$, outputs a set $\mathcal{A}'$ of size $\mathsf{size}(H,k)$ that $A_H^{\otimes k}$-represents $\mathcal{A}$, running in time $\mathsf{time}(|\mathcal{A}|, H, k)$.
    Then there exists an algorithm that, given a linear ordering of an $n$-vertex graph $G$ of width $w$, decides whether $G \to H$ in time
    \[\Oh \Big (  \big (\mathsf{size}(H,w)\cdot h + \mathsf{time}\left (\mathsf{size}(H,w) \cdot h, H, w \right ) \big ) n \Big ).\]
\end{theorem}
\begin{proof}
    Let $v_1, \dots, v_n$ be a linear ordering of $G$ of width $k$.
    For $i \in [n]$, by $E_i$ we denote the set of edges that cross the $i$-th cut, i.e., those with one endpoint in $\{v_1,\ldots,v_i\}$ and the other in $\{v_{i+1},\ldots,v_n\}$.
    For $i\in [n]$, let $X_i$ be the set that contains all vertices from $\{v_1, \dots, v_i\}$ incident to an edge from $E_i$, i.e.,
    \[
    X_i := \left\{v \in \{v_1, \dots, v_i\} \ | \ \exists v' \in \{v_{i+1}, \dots, v_n\}, vv' \in E(G) \right\}.
    \]
    Note that we have $|X_i| \leq |E_i|  \leq w$ and $X_1=\{v_1\}$ (since $G$ is connected).
    For a mapping ${c : X_i \to V(H)}$, we define the table entry $T_i[c]$ as true if there exists a  homomorphism ${\vphi : G[\{v_1, \dots, v_i\}] \to V(H)}$, such that for all $v \in X_i$ we have $\vphi(v) = c(v)$. (In other words, the keys are given by the $H$-colorings of $X_i$ and the value of the table is true if there is an extension of the coloring to the graph induced on the left-hand side of the cut.)

    This table can be easily computed in time $h^{w+1} \cdot n^{\Oh(1)}$ by the following naive dynamic programming procedure.
    We initiate every entry $T_i[c]$ to be false and every entry $T_1[c]$ to be true.
    Then, for every $i \in [2,n]$, every mapping $c' : X_{i-1} \to V(H)$, such that $T_{i-1}[c']$ is true, and every $u \in V(H)$, we check whether $c : X_{i-1}\cup\{v_i\}\to V(H)$ defined as
    \begin{equation}
        c(v) = \begin{cases}
            u      & \text{ if } v = v_i,\\
            c'(v) & v\in X_{i-1}.            
        \end{cases}  \label{eq:defc}   
    \end{equation}
    is a homomorphism from $G[X_{i-1}\cup\{v_i\}]$ to $H$.
    If so, we set $T_i[c|_{X_i}]$ to true.

    We first outline why this correctly computes the table entries (that is, that at the end $T_i[c]$ is true if and only if $c$ extends to a coloring of $G[\{v_1,\dots,v_i\}]$) and then explain how to improve on this naive algorithm. We prove the claim by induction on $i$. For $i=1$, the coloring only assigns a color to $v_1$ and does not need to be extended (and is automatically proper). Now suppose that the claim has been shown for $i=1,\dots j$ and let $\alpha:X_{j+1}\to V(H)$ be a coloring. If this extends to a coloring $\phi$ of $G[\{v_1,\dots,v_{j+1}\}]$, then $T_j[c']$ is true for $c'=\phi|_{X_j}$ (by the induction hypothesis) and we could obtain $\alpha$ as the restriction from $c$ from \cref{eq:defc} with $u=\phi(v_{j+1})$ and $i=j+1$. So $T_{j+1}[\alpha]$ is true. Vice versa, if $T_{j+1}[\alpha]$ has been set to true, then there is a $c':X_{i-1}\to V(H)$ and $u\in V(H)$ such that $c$ (again defined as in \cref{eq:defc}) is a homomorphism $G[X_j\cup \{v_{j+1}\}]\to H$ which restricts to $\alpha$ on $X_{j+1}$. By the induction hypothesis, there exists a proper coloring $\phi'$ that extends $c'$ to $G[\{v_1,\dots,v_j\}]$ and we extend this to a coloring $\phi$ of $G[\{v_1,\dots,v_{j+1}\}]$ by setting $\phi(v_{j+1})=u$. Then $\phi$ still restricts to $\alpha$ and all of the edge constraints have been verified by $c$ and/or $\phi'$. 
    
    In particular,  $G \to H$ if and only if  $T_n[\emptyset]$ is true, where $\emptyset$ denotes the empty mapping ($X_n=\emptyset$).
    We will speed up this naive version of the dynamic program by computing a representative table $T'$ as follows. We first set $T'_1 = T_1$. For $i = 1,2,\ldots,n-1$ we proceed as follows.
     Let $k=|E_i|\leq w$ and $M=A_H^{\otimes k}$. Let $\{a_1,b_1\},\dots,\{a_k,b_k\}\in E_i$ be an enumeration of the edges, with $a_j\in \{v_1,\dots,v_i\}$ for all $j\in [k]$.
     For each $c: X_i \to V(H)$ such that $T'_i[c]$ is set to true, we put the $k$-tuple $(c(a_1),\dots,c(a_k))$ in $\mathcal{A}_i$. 
    When $k\geq 2$, we apply the reduction algorithm $R$ to $\mathcal{A}_i$, resulting in a set $\mathcal{A}'_i$ of size at most $\mathsf{size}(H,k)$ that $A_H^{\otimes k}$-represents $\cA_i$. When $k=1$, we set $\mathcal{A}_i'=\mathcal{A}_i$.
    We then compute the next table entries similarly as in the previous approach. Each  element of $\mathcal{A}'_i$ corresponds to a coloring $c':X_i\to V(H)$. For $u \in V(H)$, we check whether $c : X_{i}\cup\{v_{i+1}\} \to V(H)$ with $c(v_{i+1})=u$ and $c|_{X_i}=c'$ is a homomorphism from $G[X_{i}\cup\{v_{i+1}\}]$ to $H$. If so, we set $T'_{i+1}[c|_{X_{i+1}}]=1$. We repeat this for all pairs $(c',u)$. 
    
    The procedure above is repeated for $i = 1, \ldots, n-1$, after which we return $T_n'[\emptyset]$ as the answer. 

    When $|\mathcal{A}_i'|\leq \mathsf{size}(H,k)$, we find that $|\mathcal{A}_{i+1}|\leq \mathsf{size}(H,k)h$ (for $k=|E_i|\leq w$ and $\mathsf{size}(H,k)=h$ for $k=1$). We may assume $\mathsf{size}$ is a non-decreasing function on each coordinate. So the total running time is as claimed:
    \[
    \Oh \Big (  \big (\mathsf{size}(H,w) \cdot h + \mathsf{time}\left (\mathsf{size}(H,w) \cdot h,H, w \right ) \big ) n \Big ).
    \]
    The fact that the dynamic programming steps preserve representation follows from transitivity of representation, but let us spell out the details. 

    Let $Y_{i+1}$ be the set of endpoints on the right-hand side of the $(i+1)$th cut and enumerate the edges in $E_{i+1}$ as $\{x_1,y_1\},\dots,\{x_k,y_k\}$, with $x_j\in X_{i+1}$ and $y_j\in Y_{i+1}$. We will show that for every $i\in [n-1]$, if $\mathcal{A}'_i$ represents the set $ \text{True}_i := \{(c(x_1), \dots, c(x_k)) \ |  \ T_i[c] = \text{True}\}$, then $\mathcal{A}_{i+1}$ represents the set $\text{True}_{i+1}$. 
    We started with setting $T_1'=T_1$, so $\cA_1'$ indeed represents $\text{True}_1$.
    
    Suppose that $\mathcal{A}'_i$ represents the set $\text{True}_i$ for some $i\in [n-1]$. We need to show that $\mathcal{A}_{i+1}$ represents the set $\text{True}_{i+1}$. The same then holds for $\mathcal{A}_{i+1}'$ since the reduction algorithm is assumed to work correctly.

Let us first unravel the definitions to see what we need to show.  
Let $i\in [n-1]$ and suppose that $c:G[X_{i+1}]\to H$  extends to a homomorphism $\phi:G[\{v_1,\dots,v_{i+1}\}]\to H$ (i.e. $T_i[c]=\text{True}$). For the definition of represents, we will then assume there is a homomorphism $d : G[Y_{i+1}] \to H$ for which $c\cup d$ respects all edges from the $(i+1)$th cut (those in $E_{i+1}$), i.e. this corresponds to a `one-entry in the compatibility matrix'. 
What needs to be shown is that this `one-entry' can also be generated via a coloring coming from $\mathcal{A}_{i+1}$, that is, there is $\alpha : G[X_{i+1}]\to H$, such that $(\alpha(x_1), \dots, \alpha(x_k)) \in \mathcal{A}_{i+1}$ and $(\alpha\cup d)|_{G[E_{i+1}]}$ is a homomorphism.

By assumption, $\phi\cup d$ respects all the edges with at least one endpoint in $\{v_1,\ldots,v_{i+1}\}$, and in particular those with one endpoint in $\{v_1,\dots,v_i\}$.
Since $\mathcal{A}'_i$ is a representative set of $\text{True}_i$, there must be $c' : G[X_i] \to H$ such that $(c'(x'_1), \dots, c'(x'_{k'})) \in \mathcal{A}_i'$, for $\{x'_1, \dots, x'_{k'}\} = X_i$, and where $c' \cup \phi|_{\{v_{i+1}\}} \cup d$ respects all the edges with at least one endpoint in $\{v_1,\ldots,v_{i}\}$. We set $\alpha= (c' \cup \phi|_{\{v_{i+1}\}})|_{X_{i+1}}$. Then $(\alpha(x_1), \dots, \alpha(x_k)) \in \mathcal{A}_{i+1}$, by definition of how we obtain $\mathcal{A}_{i+1}$ from $\mathcal{A}_i'$. Moreover,  $\alpha \cup d$ is a homomorphism $G[E_{i+1}]\to H$, as desired.    
\end{proof}

\section{Upper bounds and algorithms for representative sets}\label{sec:algo}
In this section we focus on upper bounds for the size of representative sets and how to actually compute $H$-representative sets, where the size guarantee for the resulting representative sets is given in terms of two different parameters of $H$.
The first of these two algorithms is one of our main technical contributions, and it is rather general since it finds representative sets non-trivially fast for any large Kronecker power of a matrix with small $\hg$ parameter. The second of these two algorithms is less innovative since it is a natural generalization of the algorithm from~\cite{DBLP:journals/tcs/JansenN19}. We compare him to the support-rank in Section \ref{subsec:hgvssuprank}.


\subsection{Computing representative sets via half-induced matchings}
In this section we show how to compute small representative sets for graphs with no large half-induced matching.
When we combine the reduction algorithm described in Lemma~\ref{lem:RepAll} with Theorem~\ref{thm:MetaAlg} we find the following result.

\begin{sloppypar}
\begin{theorem} \label{thm:HgAlg}
  The \homm problem on an instance $(G,H)$, where $G$ is given with a linear ordering of width $k$, can be solved in time $\Oh(k^{2k\cdot \hg(H)}\cdot |V(H)|^4  |V(G)|)$.
 \end{theorem}
 \end{sloppypar}

We emphasize that the algorithm does not need to know the value of $\hg(H)$.
Since $\hg(H) \leq \mimsup(H)$ (see \cref{lem:mimsupgeqhg}), we immediately obtain \cref{thm:MimAlg} as a corollary from \cref{thm:HgAlg}.

\thmMimAlg*




We will show how to find a representative set that has one fewer element, by finding some element that can be safely removed. We then use this intermediate result to find our final reduction algorithm, which will result in the following lemma.

\begin{lemma} \label{lem:RepAll}
    Let $\ell\geq 1$ and $k\geq 2$ be integers. Let $A\in \{0,1\}^{h\times h}$ be a matrix with $\hg(A) < \ell$, and let $\mathcal{A}\subseteq [h]^k$. 
    Then we can compute $\mathcal{A}'\subseteq \mathcal{A}$ that $A^{\otimes k}$-represents $\mathcal{A}$ with $|\mathcal{A}'|\leq k^{k\ell}$ in time $\Oh(|\mathcal{A}|^2h^2k^2)$.
\end{lemma}
From Lemma~\ref{lem:RepAll}, Theorem~\ref{thm:combhgbound} and Theorem~\ref{thm:HgAlg} easily follow. We postpone the details of the combinatorial bound to Section~\ref{subsec:hgvsmimsup} and for now focus on the algorithmic aspects.


\begin{proof}[Proof of Theorem~\ref{thm:HgAlg}]
    Let $h = |V(H)|$ and let $A_H$ be the adjacency matrix of $H$. Recall that $\hg(A_H)=\hg(H)$ is always an integer.
    By Lemma~\ref{lem:RepAll} we have a reduction algorithm $R$ that returns a representative set of size $\textsf{size}(H,k) \leq k^{k\cdot (\hg(H)-1)}$ in time $\textsf{time}(|\mathcal{A}|, H, k) = \Oh(|\mathcal{A}|^2h^2k^2)$. Then \[
    \textsf{time}(\textsf{size}(H,k)\cdot h,H, k) = \Oh\left(k^2 \cdot k^{2k\cdot(\hg(H)-1)} \cdot h^4\right).
    \]
    By Theorem~\ref{thm:MetaAlg} we find an algorithm that decides \homo{H} in time
    \[
    \Oh\left((\textsf{size}(H,k)\cdot h + \textsf{time}(\textsf{size}(H,k)\cdot h,H, k))|V(G)|\right) = \Oh(k^{2k\cdot\hg(H)}h^4 \cdot |V(G)|).
    \] 
    This completes the proof. 
\end{proof}
In order to prove the lemma, we will perform a recursive algorithm for which we want to no longer treat all the coordinates symmetrically. We therefore define
\[
g_k(\ell_1,\ldots,\ell_k) = \binom{\sum_i\ell_i}{\ell_1,\ldots,\ell_k}.
\]
When $\ell_1=\dots=\ell_k=\ell$, we have
 $g_k(\ell,\dots,\ell)=\binom{k\ell}{\ell,\ldots,\ell}\leq k^{k\ell}$.
The lemma will follow easily from the following more complicated statement.
\begin{lemma} \label{lem:RepOne}
Let $ k\geq 2, \ell_1,\ldots,\ell_k\geq 1$ be integers. Let $A\in \{0,1\}^{h\times h}$ be a matrix and let $\mathcal{A}\subseteq [h]^k$ with  $|\mathcal{A}|\geq g_k(\ell_1,\ldots,\ell_k)$.
Suppose that for every ${i\in [k]}$, for the set of rows $\cR_i=\{r_i \ | \ r\in \mathcal{A}\}$, we have $\hg(A[\cR_i,\cdot]) < \ell_i$.
Then there exists $v \in \mathcal{A}$ such that $\mathcal{A} \setminus \{v\}$ $A^{\otimes k}$-represents $\mathcal{A}$. Moreover, $v$ can be found in time $\Oh(\sum_{i=1}^k \ell_i \cdot |\mathcal{A}| hk)$.
\end{lemma}

\begin{proof}
Note that $|\mathcal{A}|\geq g_k(\ell_1,\dots,\ell_k)\geq 1$ for $\ell_1,\dots,\ell_k,k\geq 1$ so it is non-empty. 

For $i \in [k]$ and $u\in [h]$, let 
\[
\mathcal{A}^i_u = \{ v=(v_1,\dots,v_k)\in \mathcal{A} \ | \ A[v_i,u]=0 \}
\]
be the set of rows which cannot `represent' $u$ in the $i$th coordinate. 
We choose $v \in \mathcal{A}$ (arbitrarily). We then iterate over $u \in [h]$ and $i\in [k]$ to find if there is $(u,i)$ for which 
\begin{itemize}
    \item $A[v_i,u] =1$, and
    \item $|\mathcal{A}^i_u| \geq g_k(\ell_1,\ldots,\ell_i-1,\ldots,\ell_k)$.
\end{itemize}
This step can be performed in time $\Oh(|\mathcal{A}|hk)$.

If we cannot find such $(u,i)$ pair for $v$, then we return $v$ as the row to be removed from $\mathcal{A}$ (and the algorithm terminates).

Otherwise, we did find $(u,i)$.
If $\ell_i = 1$, then since $\hg(A[\cR_i,\cdot])<\ell_i$, we know $A[\cR_i,\cdot]$ has all zero-entries and so $A[v_i,u]=1$ would not have been possible. This means that $\ell_i\geq 2$. We apply the same process after updating $\ell_i \leftarrow \ell_i - 1$ and $\mathcal{A} \leftarrow \mathcal{A}^i_u$.
Note that $v \notin  \mathcal{A}^i_u$  and $\ell_i-1\geq 1$. 
We will show that
\begin{itemize}
    \item when $v$ is returned, indeed $\mathcal{A}\setminus\{v\}$ $A^{\otimes k}$-represents $\mathcal{A}$, and
    \item when we recursively apply the algorithm, the conditions of the lemma are again satisfied, for which it remains to show that $\hg(A[\cR_i',\cdot])<\ell_i-1$  for $\cR_i'=\{r_i\mid r \in \mathcal{A}_u^i\}$.
\end{itemize}
Since we reduce $\sum_{i=1}^k \ell_i$ by one in each recursive call, the algorithm will terminate. Moreover, the number of recursive calls is at most $\sum_{i=1}^k \ell_i$. This shows that assuming the claims above, the time complexity is as stated.

\paragraph{Correctness.}
 We first show the first claim: if the algorithm outputs $v$, indeed it can be removed. Note that when for some subset $\mathcal{A'}\subseteq\mathcal{A}$, it is the case that $\mathcal{A'}\setminus \{v\}$ represents $\mathcal{A}'$, then
 \[
 \mathcal{A'}\setminus \{v\}\cup (\mathcal{A}\setminus \mathcal{A'})=\mathcal{A}\setminus\{v\}
\]
will also represent $\mathcal{A}$. This means we only have to check the claims in the `base case'. 
Suppose towards a contradiction that we wrongly outputted $v\in \mathcal{A}$, so
\begin{itemize}
    \item there exists $u\in [h]^{ k}$ such that $A^{\otimes k}[v,u]=1$ yet $A^{\otimes k}[v',u]=0$ for all $v'\in \mathcal{A}\setminus\{v\}$ (since we `wrongly' outputted $v$, there needs to be a reason why we could not remove it),
    \item for this $u$, for all $i \in [k]$, $|\mathcal{A}_{u_i}^i|<g_k(\ell_1,\dots,\ell_i-1,\dots,\ell_k)$ (else the algorithm would have `recursed' instead of outputting $v$).
\end{itemize}
The fact that $A^{\otimes k}[v',u]=0$ in the first condition, means that each $v'\in \mathcal{A}\setminus\{v\}$ is an element of $\mathcal{A}_{u_i}^i$ for some $i\in [k]$. In particular,
\[
|\mathcal{A} \setminus \{v\}| \leq \sum \limits_{i=1}^k |\mathcal{A}^i_{u_i}|\leq \sum_{i=1}^k g_k(\ell_1,\ldots,\ell_i-1,\ldots,\ell_k)-k = g_k(\ell_1,\ldots,\ell_k)-k,
\]
which contradicts the assumptions of the lemma since $k\geq 2$.

We now prove the second claim: the conditions of the lemma are satisfied when we `recurse'. By assumption, $\ell_i \geq 1$ for all $i$ and the new $\mathcal{A}$ is sufficiently large. Moreover, $\hg$ can also decrease when taking submatrices, so indeed we only need to show that $\hg(A[\cR_i',\cdot])<\ell_i-1$ for $\cR_i'=\{r_i\mid r\in \mathcal{A}_{u_i}^i\}$. If there is a half-induced matching of size $\ell_i-1$, induced on rows $w_1,\ldots,w_{\ell_i-1}\in \cR_i'$ and columns $z_1,\ldots,z_{\ell_i-1}$, then there is a half-induced matching of size $\ell_i$ in $A[\cR_i,\cdot]$ by considering rows $w_1,\ldots,w_{\ell_i-1},v_i\in \cR_i$ and columns $z_1,\ldots,z_{\ell-1},u$. But by assumption this does not exist, so indeed $\hg(A[\cR_i',\cdot])<\ell_i-1$.
\end{proof}

\begin{proof}[Proof of \cref{lem:RepAll}]
Suppose that $|\mathcal{A}| \geq g_k(\ell, \dots, \ell)$. For $i\in [k]$, set $\cR_i = \{r_i \ | \ r \in \mathcal{A}\}$. Then $\hg(A[\cR_i,\cdot]) < \ell$ for each $i$. By Lemma \ref{lem:RepOne} we can find a row $v$  in $\cA$ such that $\mathcal{A} \setminus \{v\}$ $A^{\otimes k}$-represents $\mathcal{A}$
in time $\Oh(\ell k \cdot |\mathcal{A}| hk) = \Oh(|\mathcal{A}| h^2 k^2)$, where we use that $\ell \leq h$.

We repeat this at most $|\mathcal{A}| - g_k(\ell,\ldots,\ell)$ times until we find the desired representative set in time $\Oh(|\mathcal{A}|^2h^2k^2)$.
\end{proof}

\subsection{Computing representative sets via support rank}
\label{sec:suprank}
The \textit{support-rank} (called e.g. `non-deterministic rank' in~\cite{doi:10.1137/S0097539702407345}, see also~\cite{DBLP:conf/birthday/Nederlof20}) of a matrix $M\in \{0,1\}^{n\times m}$ over the field $\F$ is defined as
\[
\text{support-rank}(M)=
\min\{\rank_\F(M') \ | \ M'\in \F^{n\times m}\text{ and }M_{ij}=0\iff M_{ij}'=0\}.
\]

The proof of Jansen and Nederlof~\cite{DBLP:journals/tcs/JansenN19} builds on the fact that support-rank of the adjacency matrix of $K_q$ is at most $2$, since the sum of two rank one matrices
\[
\begin{pmatrix}
    1 & 1 &\dots &1\\
    2& 2& \dots &2\\
    \vdots & & \ddots & \vdots\\
    n & n&\dots  & n\\
\end{pmatrix}- \begin{pmatrix}
    1 & 2 &\dots &n\\
    1& 2& \dots &n\\
    \vdots & & \ddots & \vdots\\
    1 &2&\dots & n\\
\end{pmatrix}
\]
has the same non-zero entries.
A small support-rank allows representative sets to be computed efficiently using row elimination. 
We record here what a generalization of this approach would give for our setting. 

The following lemma has no new ideas, but is a generalization of the approach of Jansen and Nederlof \cite{DBLP:journals/tcs/JansenN19}.

\begin{lemma}
\label{lem:support_rank}
Let $\mathbb{F}$ be any field.\footnote{We assume that all arithmetic operations over $\mathbb{F}$ are performed in constant time.}
Let $A\in \{0,1\}^{h\times h}$ be a matrix and let $B\in \mathbb{F}^{h\times h}$ be a matrix with the same support as $A$ (i.e., $B_{i,j}=0\iff A_{i,j}=0$). For any $\ell\in \mathbb{N}$ and set of rows $\mathcal{R}\subseteq [h]^\ell$ of $A^{\otimes \ell}$, we can find a representative set $\mathcal{R}'$ for $\mathcal{R}$ (with respect to $A^{\otimes \ell}$) of size at most $\textup{rank}(B)^{\ell}$ in time  $\Oh(|\mathcal{R}|\textup{rank}(B)^{\ell\cdot (\omega-1)}\cdot \ell+h^3)$.
\end{lemma}
Here by $\omega \leq 2.3728596$ we denote the matrix multiplication exponent~\cite{DBLP:conf/stoc/Williams12}.
\begin{proof}[Proof of \cref{lem:support_rank}]
We may assume that $|\mathcal{R}|>\rank(B)^\ell$, since otherwise we are done.

We use that by the proof of \cite[Lemma 3.15]{DBLP:journals/iandc/BodlaenderCKN15} we can compute a row basis for an $n\times m$ matrix with $m\leq n$ and entries in $\mathbb{F}$ in time $\Oh(nm^{\omega -1})$. (Let us point out that  the lemma in the paper is written in the setting that $\mathbb{F} = \mathbb{F}_2$, but this assumption is actually unnecessary.)

We wish to compute a row basis $\mathcal{R}'$ for the matrix $B^{\otimes \ell }[\mathcal{R},\cdot]$, but that matrix has too many columns. 
By standard linear algebra, we can write $B=LR$ for $L$ of dimensions $h\times \rank(B)$ and $R$ of dimensions $\rank(B)\times h$ -- this can be done in time $\Oh(h^{3})$ by Gaussian elimination.
Then $B^{\otimes \ell}=L^{\otimes \ell}R^{\otimes \ell}$ and so a row basis for $L^{\otimes \ell }[\mathcal{R},\cdot]$ is also a row basis for $B^{\otimes \ell }[\mathcal{R},\cdot]$.
Note that we also cannot permit ourselves to compute $L^{\otimes \ell}$ since this has too many rows again. However, we can compute an entry $L^{\otimes \ell}[x,y]=L[x_1,y_1]L[x_2,y_2]\cdots L[x_{\ell},y_{\ell}]$. This allows us to compute the $|\mathcal{R}|\times \rank(B)^\ell$ matrix $L'=L^{\otimes \ell}[\mathcal{R},\cdot]$ in time $\Oh(|\mathcal{R}|\rank(B)^\ell \cdot \ell)$.
Since $|\mathcal{R}|>\rank(B)^\ell$, we can now compute a row basis $\mathcal{R}'$  for $L'$ in time $\Oh(|\mathcal{R}|\rank(B)^{\ell\cdot(\omega-1)})$. This is then also a row basis for $B^{\otimes \ell }[\mathcal{R},\cdot]$.

We claim such a row basis forms the desired representative set. 
    Suppose that $A^{\otimes \ell}[r,c]=1$ for some row $r\in \mathcal{R}$ and column $c\in [h]^{\ell}$ of $A^{\otimes \ell}$. We need to prove that $A^{\otimes \ell}[r',c]=1$ for some $r'\in \mathcal{R}'$. 
    Since $\mathcal{R}'$ is a row basis, there exist coefficients $a_{r'}\in \mathbb{F}$ such that
    \[
    B^{\otimes \ell}[r,c] =\sum_{r'\in \mathcal{R}}a_{r'} B^{\otimes \ell}[r',c].
    \]
    Combined with the fact that $B$ has the same support as $A$, we find \[
    A^{\otimes \ell}[r,c]\neq 0\implies B^{\otimes \ell}[r,c]\neq 0 \implies B^{\otimes \ell}[r',c]\neq 0 \text{ for some }r' \implies A^{\otimes \ell}[r',c]\neq 0 \text{ for some }r'.
    \]
    This shows that indeed $\mathcal{R'}$ is a representative set for $\mathcal{R}$. Moreover, the size of the row basis is at most the rank of $B^{\otimes \ell}$, which is $\text{rank}(B)^\ell$ as claimed.
\end{proof}

Combining \cref{lem:support_rank} with \cref{thm:MetaAlg}, we immediately obtain the following.

\begin{corollary}\label{cor:alg-supportrank}
    Let $H$ be a non-bipartite graph on $h$ vertices.
    Suppose we are given an $h \times h$ matrix over a field $\mathbb{F}$ with the same support as the adjacency matrix of $H$ and rank $r$.
     Then there exists an algorithm that, given a linear ordering of an $n$-vertex graph $G$ of width $k$, decides whether $G \to H$ in time
    $\Oh \Big ( (r^{k \cdot \omega} hk + h^3)|V(G)|\Big ).$
\end{corollary}
We saw in the previous section that there is always an $A^{\otimes k}$-representative set of size at most $\mimsup(A)^k$.  The size constraint of \cref{lem:support_rank} is possibly worse, but in this case we can guarantee that the representative sets can also be computed efficiently.  We remark that Lemma \ref{lem:support_rank} in particular shows that the support-rank is an upper bound on mimsup. This can also be seen directly and we record this fact in the following observation.
\begin{observation}
\label{obs:supprankvsmimsup}
    Let $A\in \{0,1\}^{n\times n}$ be a $0/1$-matrix and let $\mathbb{F}$ be a field. Let $B\in \mathbb{F}^{n\times n}$ be a matrix with the same support as $A$. Then $\mim(A^{\otimes k})\leq \rank_\mathbb{F}(B^{\otimes k})=\rank_\mathbb{F}(B)^k$ for every $k\in \mathbb{N}$. In particular, \[
    \mimsup(A)\leq \textup{support-rank}(A).
    \]
\end{observation}
We do not know whether mimsup and support-rank are functionally equivalent, but we do provide a separation between support-rank and him in Section \ref{subsec:hgvssuprank}.

\subsection{Bounding support rank via local biclique covers}
The caveat in \cref{cor:alg-supportrank} is that a small-rank matrix with the same support as the adjacency matrix of $H$ must be given.
If $\F$ is a finite field, an optimal such matrix can be found in time $|\F|^{h^2} \cdot h^{\Oh(1)}$ by brute-force, which is constant if both $|\F|$ and $h$ are constants.
We will now present a `combinatorial' approach for finding a small-rank matrix with the same support, which does not necessarily achieve the support-rank but which can be computed efficiently.

Let $F$ be a bipartite graph with bipartition classes $X,Y$. By $F^c$ we denote the \emph{bipartite complement} of $F$, i.e., the bipartite graph with bipartition classes $X,Y$, where $uv \in X \times Y$ is an edge if and only if $uv\notin E(F)$.

For a bipartite graph $F$, let $\cB=\{B_1,\ldots,B_s\}$ be a family of subgraphs of $F$, such that (i) each $B_i$ is a biclique, (ii) $\bigcup_{i=1}^s E(B_i)=E(F)$, and (iii) every $v\in V(F)$ is in at most $r$ bicliques of $\cB$. Then we say that $\cB$ \emph{$r$-covers} $F$.
The minimum $r$ for which there exists a family that $r$-covers $F$, has been studied under names \emph{bipartite degree}, \emph{local biclique cover number}~\cite{FISHBURN1996127,DBLP:journals/gc/DongL07}, and is as special case of so-called \emph{local covering numbers} also studied in the literature~\cite{BUJTAS2020103114,KNAUER2016745}.

\begin{lemma}\label{lem:suprank-dim}
   Let $H$ be a non-bipartite graph and assume we are given a family $\cB$ of bicliques that $r$-covers the bipartite complement $(H^*)^c$ of $H^*$. Then we can compute a matrix $A_H'$ with the same support as the adjacency matrix $A_H$ of $H$ with $\rank_{\mathbb{R}}(A'_H)\leq(r+1)^r$ in time $\Oh(h^2r^2)$. In particular the support rank of $A_H$ is at most $(r+1)^r$.
\end{lemma}

\begin{proof}
   Define an arbitrary ordering $B_1,B_2,\ldots,B_s$ of the elements of $\cB$. 
   Moreover, for every vertex of $V(H^*)$, we fix the ordering of bicliques containing it.
   For each $v \in V(H)$ and $i \in [r]$, define $\sigma_i(v) = p$ (resp. $\delta_i(v)=p$) if the $i$-th biclique covering $v'$ (resp. $v''$) is $B_p$. If $v'$ (resp. $v''$) is covered by $r'<r$ bicliques, then for $i=r'+1,\ldots,r$, we define $\sigma_i(v)=s+1$ (resp. $\delta_i(v)=s+2$).
   For $u, v\in V(H)$, we define:
   \[
   A_H'[u,v]=\prod_{i=1}^r \prod_{j=1}^r (\sigma_i(u)-\delta_j(v)).
   \]
   Observe that such a product is non-zero if and only if there is no biclique from $\cB$ that contains both $u',v''$, and this in turn happens if and only if $u'v''\in E(H^*)$ which is equivalent to $uv\in E(H)$. Therefore, $A_H'$ has the same support as $A_H$. We claim that the rank of $A_H'$ is at most $(r+1)^r$.

   In order to bound the rank of $A_H'$ we can rewrite
  \begin{align*}
 A_H'[u,v]= & \ \prod_{i=1}^{r} \Big(\sum_{\ell=0}^{r} \sigma_i(u)^{\ell}\cdot \sum_{J\subseteq [r] \ : \  |J|=r-\ell}\prod_{j\in J} (-\delta_j(v))\Big) \\ & = \sum_{(\ell_1,\ldots,\ell_{r}) \ : \ \ell_i\in [r]\cup\{0\}}  \prod_{i=1}^{r} \sigma_i(u)^{\ell_i} \cdot \prod_{i=1}^{r} \sum_{J_i\subseteq [r] \ : \  |J_i|=r-\ell_i}\prod_{j_i\in  J_i} (-\delta_{j_i}(v)).
\end{align*}

If for $i\in [r], \ell_i\in [r]\cup \{0\}$  we define
\begin{align*}
  &  L[u,(\ell_1,\ldots,\ell_r)]=\prod_{i=1}^{r} \sigma_i(u)^{\ell_i}, \\
  &  R[(\ell_1,\ldots,\ell_r),v]=\prod_{i=1}^{r} \sum_{J_i\subseteq [r] \ : \  |J_i|=r-\ell_i}\prod_{j_i\in  J_i} (-\delta_{j_i}(v)),
\end{align*}
then we see that $A_H'$ is the product of two matrices $L,R$ such that number of columns of $L$ and number of rows of $R$ is $(r+1)^r$. Therefore both $L,R$ have rank at most $(r+1)^r$ and since the rank of the product is at most the minimum of the ranks of factors, we conclude that $\rank(A_H')\leq (r+1)^r$, which completes the proof.
\end{proof}

For a non-bipartite graph $H$, let $\mathrm{cov}(H)$ denote the minimum $r$ for which there exists a family that $r$-covers $(H^*)^c$.
Note that if $H$ is assumed to be fixed, i.e., in the setting of the \homo{H} problem, the value of $\mathrm{cov}(H)$ and the actual covering family can be computed in constant time by brute force. Thus, combining \cref{lem:suprank-dim} and \cref{cor:alg-supportrank} we obtain the following.

\begin{corollary}\label{cor:alg-support-biclique}
    Let $H$ be a fixed non-bipartite graph and let $r = \mathrm{cov}(H)$.
    The \homo{H} problem on $n$-vertex instances given with a linear ordering of width $k$ can be solved in time
    $\Oh \Big ( (r+1)^{rk \cdot \omega}n^2\Big ).$
\end{corollary}

\section{Cores and prime factorizations: preprocessing $H$}\label{sec:algofinal}
In this section we show that the algorithms discussed in \cref{sec:algo}, in particular the result from \cref{thm:MimAlg}, can be improved by some standard preprocessing steps that simplify $H$.

In what follows we assume that $H$ is non-bipartite and has no vertices with loops; recall that otherwise $\homo{H}$ can be solved in polynomial time~\cite{DBLP:journals/jct/HellN90}.

\paragraph{Cores.}
A graph $H$ is a \emph{core} if it does not admit a homomorphism to any of its proper subgraphs. A \emph{core} of $H$ is a subgraph $C$ of $H$, such that $C$ is a core and there is a homomorphism from $H \to C$. 
As shown by Hell and Ne\v{s}et\v{r}il~\cite{DBLP:journals/dm/92}, every graph has a unique core (up to isomorphism). Thus we can talk about \emph{the} core of a graph $H$ and denote it by $\core(H)$. It is straightforward to verify that $\core(H)$ is actually an induced subgraph of $H$.
Furthermore, for non-bipartite graphs $H$, $\core(H)$ is non-bipartite.

Observe that homomorphisms are transitive: if $G_1 \to G_2$ and $G_2 \to G_3$, then $G_1 \to G_3$.
Furthermore, for every graph $H$ we have $\core(H) \to H$. Indeed, the function mapping each vertex of $\core(H)$ to itself is clearly a homomorphism.
Consequently, $G$ admits a homomorphism to $H$ if and only if it admits a homomorphism to $\core(H)$. Thus, in order to solve $\homo{H}$, one can equivalently focus on solving $\homo{\core(H)}$.

\begin{sloppypar}
As $\core(H)$ is an induced subgraph of $H$, by \cref{obs:mimsup-subgraph} we obtain that $\mimsup(\core(H)) \leq \mimsup(H)$. Thus the running time from \cref{thm:MimAlg} can be improved to $\exp (\mimsup(\core(H)) k \log k) \cdot (|V(G)| + |V(H)|)^{\Oh(1)}$, assuming that $\core(H)$ is given with $H$ (finding the core of a graph is a computationally hard task~\cite{DBLP:journals/dm/92}).
However, there is one more observation that we can use.
\end{sloppypar}

\paragraph{Direct products.}
For graphs $H_1, H_2$, their \emph{direct product} is the graph $H_1 \times H_2$ defined as follows
\begin{align*}
    V(H_1 \times H_2) = & \{ (v_1,v_2) ~|~ v_1 \in V(H_1) \text{ and } v_2 \in V(H_2) \}\\
    E(H_1 \times H_2) = & \{ (v_1,v_2)(u_1,u_2) ~|~ v_1u_1  \in E(H_1) \text{ and } v_2u_2 \in E(H_2) \}.
\end{align*}
Graphs $H_1$ and $H_2$ are \emph{factors} of $H$ and $H_1 \times H_2$ is a \emph{factorization} of $H$ (formally, a factorization is a sequence of factors).
These definitions can be naturally generalized to more factors.
A graph $H$ with at least two vertices is \emph{prime} if it cannot be written as a direct product of at least two graphs, each with at least two vertices.
A factorization, where each factor is prime and has at least two vertices is called a \emph{prime factorization}.

The following statement follows from a well-known result of McKenzie~\cite{McKenzie1971}; see also~\cite[Theorem 8.17]{hammack2011handbook}.
\begin{theorem}[McKenzie~\cite{McKenzie1971}]\label{thm:unique-fact}
Any connected non-bipartite core has a unique (up to reordering of factors) prime factorization. Furthermore, such a factorization can be found in polynomial time.
\end{theorem}

The fact why direct products play an important role in the study of graph homomorphisms is the following straightforward observation: if $H = H_1 \times H_2 \times \ldots \times H_p$ for some graphs $H_1,\ldots,H_p$, then $G \to H$ if and only if $G \to H_i$ for every $i \in [p]$. So we can solve the problem for each $H_i$ independently and then return the conjunction of answers.
Combining this with the previous observation about homomorphisms to cores immediately yields \cref{thm:mainalgohomo}.

 \begin{restatable}{theorem}{mainalgohomo}
 \label{thm:mainalgohomo}
 Let $H$ be a non-bipartite connected graph given along with its core $\core(H)$.
 Let $H_1 \times \ldots \times H_p$ be the prime factorization of $\core(H)$.
 Let $c_H = \max_{i \in [p]} \mimsup(H_i)$.
 Then there exists an algorithm that, given $H$ and a graph $G$ with a linear ordering of width $k$, decides whether $G \to H$ in time
    $k^{2c_H \cdot k} \cdot (|V(G)|+|V(H)|)^{\mathcal{O}(1)}$.
\end{restatable}
 Let us emphasize that the improvements in \cref{thm:mainalgohomo} over \cref{thm:MimAlg} are performed ``outside'' the dynamic programming procedure so they should not be treated as an evidence that the bound from \cref{thm:reduce} is not tight.

The statement of \cref{thm:mainalgohomo} becomes much cleaner if we add some (fairly natural) assumptions on the graph $H$.
Observe that for each $H$, each $k \geq 2$, and each $i \in [k]$, the projection to the $i$-th coordinate defines a homomorphism from $H^{\times k}$ to $H$, where $H^{\times k}$ denotes the direct product of $k$ copies of $H$.
A graph $H$ is \emph{projective} if \emph{every} homomorphism from $H^{\times k}$ to $H$, for every $k \geq 2$, is a projection (possibly combined with an automorphism of $H$).

Let us now assume that $H$ is a non-bipartite projective connected core.
Such an assumption might seem artificial, but, as we already mentioned in the introduction, almost all graphs are non-bipartite projective cores.
It is worth to mention that projective graphs are always connected and prime~\cite{larose2001strongly}.
Thus for this class of graphs we obtain that the following result.

 \begin{restatable}{corollary}{alghomoprojective}
 \label{cor:algohomoprojective}
 There exists an algorithm that, given a non-bipartite projective core $H$ and a graph $G$ with a linear ordering of width $k$, decides whether $G \to H$ in time     $k^{2\mimsup(H)\cdot k} \cdot (|V(G)|+|V(H)|)^{\mathcal{O}(1)}$.
 \end{restatable}

The reader might wonder why we assume that $H$ is projective, and not only that it is prime. This will become more clear in \cref{sec:lower},
when we can nicely complement \cref{cor:algohomoprojective} with strong lower bounds.

\paragraph{Further improvements?} As already highlighted in the introduction, the assumption that $H$ is projective (and a core) is not as restrictive as it might seem at the first glance. Indeed, recall that \emph{almost every graph} is a projective core~\cite{DBLP:journals/dm/92,luczak2004note}.
However, there is still a question whether there are some non-projective graphs that are interesting for us (in particular, are non-bipartite cores).
We discuss it in Section~\ref{sec:conc}.

\section{Lower bounds for \homo{H} parameterized by cutwidth}
\label{sec:lower}
It will be convenient to state our lower bounds for the \emph{list homomorphism problem}, \lhomo{H}, first as intermediate result. Let us start with some basic definitions and notation.

For a fixed graph $H$, an instance of the \lhomo{H} problem is a pair $(G,L)$, where $G$ is a graph and $L : V(G) \to 2^{V(H)}$ is a \emph{list function} (we sometimes say  $L$ are \emph{$H$-lists} if we want to emphasize the image).
We ask whether there exists a homomorphism $\vphi : G \to H$ which additionally respects lists $L$, i.e., for every $v \in V(G)$ it holds that $\vphi(v) \in L(v)$.
We denote this shortly by writing $\vphi : (G,L) \to H$.
If we just want to indicate that such a $\vphi$ exists, we simply write $(G,L) \to H$.

Note that in such a setting it makes sense if vertices of $H$ have loops,
but for the scope of this paper we will focus on simple graphs.
In such a case \lhomo{H} is polynomial-time solvable if $H$ is bipartite and its complement is a circular-arc graph, and \textsf{NP}-complete otherwise~\cite{DBLP:journals/combinatorica/FederHH99}.

In our argument, a special role is played by the case if $H$ is bipartite.
For a connected, bipartite graph $H$ with bipartition classes $X,Y$, we say that $(G,L)$ is a \emph{consistent instance} of $\lhomo{H}$, if the following conditions are satisfied:
\begin{enumerate}
\item $G$ is connected and bipartite with bipartition classes $X_G,Y_G$,
\item $\bigcup_{v\in X_G} L(v) \subseteq X$ and  $\bigcup_{v\in Y_G} L(v) \subseteq Y$.
\end{enumerate}
We also say that a subset of $V(H)$ is \emph{one-sided} if it is contained either in $X$ or in $Y$.

Let us introduce one more term. Two vertices are \emph{incomparable} if their neighborhoods are not contained in each other. A set of vertices is \emph{incomparable} if all vertices in this set are pairwise incomparable.
Finally, a graph $H$ is \emph{incomparable} if $V(H)$ is an incomparable set.

The main technical ingredient is the following lower bound.

\begin{restatable}{theorem}{ListBipOrd}\label{thm:list-bip-ord}
Let $\cH_{0}$ denote the set of connected incomparable bipartite graphs whose complement is not a circular-arc graph.
\begin{enumerate}[(1.)]
\item Assuming the ETH, there exists $\delta>0$ such that for every $H \in \cH_{0}$ the following holds. 
There is no algorithm that solves every consistent instance $(G,L)$ of $\lhomo{H}$, given with a linear ordering of $V(G)$ of width $t$,
in time $\mimsup(H)^{\delta\cdot t}\cdot |V(G)|^{\Oh(1)}$.
\item Assuming the SETH, for every $\epsilon >0$ and $H \in \cH_0$ the following holds.
There is no algorithm that solves every consistent instance $(G,L)$ of $\lhomo{H}$, given with a linear ordering of $V(G)$ of width $t$, in time $(\mimsup(H)-\eps)^{t}\cdot |V(G)|^{\Oh(1)}$.
\end{enumerate}
\end{restatable}

We postpone the proof of \cref{thm:list-bip-ord}, and first let us show how it implies the lower bounds for \homo{H}. 

\subsection{Proof of \cref{thm:lower}, assuming \cref{thm:list-bip-ord}}

As a first step, we will use \cref{thm:list-bip-ord} to show hardness of \lhomo{H} for non-bipartite incomparable graphs $H$.
Then we will extend the construction and show hardness of \homo{H}.

\subsubsection{Hardness of \lhomo{H}, non-bipartite case.}

We will use the following simple observation.

\begin{proposition}[\cite{DBLP:conf/esa/OkrasaPR20}]\label{prop:H-star}
Let $H$ be a graph, and let $(G,L)$ be a consistent instance of $\lhomo{H^*}$. Define $L': V(G)\to 2^{V(H)}$ as $L'(v)=\{u \ | \ \{u',u''\}\cap L(v)\neq \emptyset \}$. Then $(G,L)\to H^*$ if and only if $(G,L')\to H$.
\end{proposition}

Combining \cref{thm:list-bip-ord} with \cref{prop:H-star} yields the following lower bound.

\begin{restatable}{theorem}{ListGenOrd}\label{thm:list-gen-ord}
Let $\cH_{1}$ denote the set of connected incomparable non-bipartite graphs.
\begin{enumerate}[(1.)]
\item Assuming the ETH, there exists $\delta>0$ such that for every $H \in \cH_{1}$ the following holds. 
There is no algorithm that solves every instance $(G,L)$ of $\lhomo{H}$, given with a linear ordering of $V(G)$ of width $k$,
in time $\mimsup(H)^{\delta\cdot k}\cdot |V(G)|^{\Oh(1)}$.
\item Assuming the SETH, for every $\epsilon >0$ and $H \in \cH_1$ the following holds.
There is no algorithm that solves every instance $(G,L)$ of $\lhomo{H}$, given with a linear ordering of $V(G)$ of width $k$, in time $(\mimsup(H)-\eps)^{k}\cdot |V(G)|^{\Oh(1)}$.
\end{enumerate}
\end{restatable}
\begin{proof}
Fix any $H \in \cH_1$.
Observe that $H^* \in \cH_0$, where $\cH_0$ is defined as in \cref{thm:list-bip-ord}. Indeed, the connectivity of $H^*$ follows easily from the fact that $H$ is connected and non-bipartite, see e.g.~\cite[Observation 2.5]{DBLP:journals/siamcomp/OkrasaR21}.
The fact that $H^*$ is incomparable follows from the fact that $H$ is incomparable and has no isolated vertices.
Finally, $H$ contains an (induced) odd cycle, so $H^*$ contains an induced cycle with at least 6 vertices, which is an obstruction to being the complement of a circular-arc graph~\cite{DBLP:journals/combinatorica/FederHH99}.

Now suppose that we have an algorithm $\cA$ that solves every instance $(\widetilde{G},\widetilde{L})$ of $\lhomo{H}$ in time $f(H,\widetilde{G})$, where $f$ is some function that depends on $H$ and  $\widetilde{G}$.
\cref{prop:H-star} implies that that for any consistent instance $(G,L)$ of $\lhomo{H^*}$, we can solve it by calling $\cA$ on the instance $(G,L')$ of $\lhomo{H}$, defined as in \cref{prop:H-star}, in time $f(H,G)$.
Recall that for non-bipartite $H$, we have $\mimsup(H) = \mimsup(H^*)$.
Thus the statement of the theorem follows directly from \cref{thm:list-bip-ord}.
\end{proof}

\subsubsection{Hardness of \homo{H}}
Now we extend our hardness results to the $\homo{H}$ problem; the proof follows the ideas of~\cite{DBLP:conf/stacs/PiecykR21}.
Recall from \cref{sec:algofinal} that we can safely assume that $H$ is a core,
as $\homo{H}$ and $\homo{\core(H)}$ problems are equivalent.
We will require the following straightforward properties of cores, see e.g.~\cite{DBLP:journals/siamcomp/OkrasaR21}.

\begin{proposition}\label{prop:cores}
Let $H$ be a core. 
    \begin{enumerate}[(1)]
        \item Every homomorphism $\vphi: H\to H$ is an automorphism.
        \item The graph $H$ is incomparable.
        In particular, the neighborhoods of vertices in $H$ are pairwise distinct.
    \end{enumerate}
\end{proposition}

We will reduce from $\lhomo{H}$ to $\homo{H}$.
The idea is to start with an instance with lists $L$, and, for each vertex $v$, introduce a constant-sized gadget that ``simulates''  $L(v)$.
This ``simulation'' is formalized by the notion of \emph{constructions}. 
Let $H$ be a graph. For a set $S \subseteq V(H)$, its construction consists of:
\begin{itemize}
\item a graph $\textsc{C}(S)$,
\item a tuple $(x_1,x_2,\ldots,x_\ell)$ of vertices of $H$,
\item a tuple $(y_1,y_2,\ldots,y_\ell)$ of vertices of $\textsc{C}(S)$,
\item one special vertex $y_0$ of $\textsc{C}(S)$,
\end{itemize}
such that the following property is met:
\[\{\vphi(y_0) \ | \ \vphi: \textsc{C}(S) \to H  \text{ such that } \vphi(y_1)=x_1,\ldots,\vphi(y_\ell)=x_\ell\} = S.
\]
In other words, if we map each vertex $x_i$ to its corresponding vertex $y_i$, 
then the set of  possible images of $y_0$ over all extensions of this partial mapping to a homomorphism from $\textsc{C}(S)$ to $H$ is exactly $S$.
We will use the following characterization of projective graphs.

\begin{theorem}[Larose, Tardif~\cite{larose2001strongly}]\label{thm:projective}
A graph $H$ with at least three vertices is projective if and only if every $S \subseteq V(H)$ has a construction.
\end{theorem}

Now we are ready to prove the following.

\thmlower*
\begin{proof}
    We will reduce from $\lhomo{H}$ using lower bounds of \cref{thm:list-gen-ord}. Let $(G,L)$ be an instance of $\lhomo{H}$ given along with a linear ordering $\sigma = (v_1,v_2,\ldots,v_{|V(G)|})$ of $V(G)$ of width $k$. We will construct, in time polynomial in $|V(G)|$, an instance $\tG$ of $\homo{H}$ with the following properties:
    \begin{enumerate}
        \item $\tG\to H$ if and only if $(G,L)\to H$,
        \item $|V(\tG)|=|V(G)|\cdot f(H)$,
        \item $\ctw(\tG)\leq k+g(H)$, 
    \end{enumerate}
    where $f$ and $g$ are functions whose value depends only on $H$.

   We start constructing $\tG$ by taking a copy of $G$. Then for every vertex $v\in V(G)$ we proceed as follows. Using \cref{thm:projective}, we introduce a copy of the construction of $L(v)$, i.e., for $S=L(v)$, we introduce the graph $\mathsf{C}(S)$,
tuples $(x_1,\ldots,x_\ell) \in V(H)^{\ell}$ and
$(y_1,\ldots,y_\ell) \in V(\mathsf{C}(S))^{\ell}$,
and the vertex $y_0 \in V(\mathsf{C}(S))$.
We identify $y_0$ with $v$.

Next, we introduce a copy $H^v$ of $H$; for any $z \in V(H)$, let $z^v$ denote the copy of $z$ in $H^v$. 
For each $i \in [\ell]$ we identify $x_i^v$ with $y_i$.
Finally, for every pair $v_j,v_{j+1}\in V(G)$, we proceed as follows.
For every $wz\in E(H)$, we add an edge between $w^{v_j}$ and $z^{v_{j+1}}$.
This completes the construction of $\tG$. 
Clearly the construction is performed in time polynomial in $|V(G)|$,
and $|V(\tG)|=|V(G)|\cdot f(H)$ for some function $f$ that depends only on $H$.

Let us verify the equivalence of the instances $\tG$ and $(G,L)$.
First suppose that there exists $\vphi: (G,L)\to H$. We define $\vphi': \tG \to H$ as follows.
We set $\vphi'|_{V(G)}=\vphi$ and for every copy of $H$ we define $\vphi'$ on its vertices as the identity function.
Note that the function defined so far respects the edges inside $V(G)$ since $\vphi$ is a homomorphism, and the edges between copies of $H$ as vertices of two consecutive copies of $H$ are adjacent only to vertices that correspond to their neighbors in their own copy.
It remains to extend $\vphi'$ to the remaining vertices of graphs $\mathsf{C}(S)$. We can do it independently for every graph $\mathsf{C}(S)$ as there are no edges between them. Consider such $\mathsf{C}(S)$ introduced for some $v\in V(G)$ and $S = L(v)$. 
Since $\vphi'(y_0)=\vphi'(v)\in L(v)$ and $\vphi(y_i)=x_i$ for every $i\in [\ell]$,
\cref{thm:projective} implies that the mapping $\vphi'$ can indeed be extended to a homomorphism from $\tG$ to $H$.

For the other direction, suppose that there is $\vphi': \tG \to H$.
We define $\vphi=\vphi'|_{V(G)}$. Clearly $\vphi$ is a homomorphism.
Let us verify that it respects lists.
First, since $H$ is a core, \cref{prop:cores}~(1) implies that $\vphi'$ restricted to any copy of $H$ is an automorphism.
We claim that for each copy it is actually the same automorphism of $H$.
Indeed, this follows easily from the fact that the vertices of $H$ have pairwise distinct neighborhoods (by \cref{prop:cores}~(2)) and the way how any two consecutive copies of $H$ are connected.
Without loss of generality assume that this ``global'' automorphism of $H$ is the identity.
Now consider any $v \in V(G)$. Since each vertex $x_i$ from the copy of $\mathsf{C}(L(v))$ corresponding to $v$ is mapped to $y_i$ and $\mathsf{C}(L(v))$ is a construction of $L(v)$, we conclude that $v (=y_0)$ is mapped to some element of $L(v)$. Thus $\vphi'$ respects lists.

Now let us define a linear ordering $\widetilde{\sigma}$ of $V(\tG)$. First we order the vertices of $G$ according to $\sigma$.
Then we modify this ordering by inserting right after a vertex $v$ the vertices of $H^v$ and the graph $\mathsf{C}(L(v))$ introduced for $v$; 
the order among these vertices is arbitrary.
This completes the definition of $\widetilde{\sigma}$. 
Consider any cut in $\widetilde{\sigma}$. Among the edges crossing this cut there can be:
\begin{enumerate}[a)]
\item at most $k$ edges of $E(G)$,
\item edges from at most one copy of $H$ and at most one graph $\mathsf{C}(S)$,
\item edges joining two consecutive copies of $H$.
\end{enumerate}
Observe that the number of edges in b) and c) can be bounded by a constant that depends only on $H$, say $g(H)$. Therefore, we can conclude that $\ctw(\tG)\leq k+ g(H)$.

Now consider $\delta$ from \cref{thm:list-gen-ord} and suppose there is an algorithm $\cA_1$ that solves every instance $G'$ of $\homo{H}$ in time $\mimsup(H)^{\delta\cdot\ctw(G')}\cdot |V(G')|^{\Oh(1)}$.
Then, we can call the above reduction for any instance $(G,L)$ of $\lhomo{H}$ given along with a linear ordering of $V(G)$ of width $k$ and use $\cA_1$ to solve it in time:
   \[
   \mimsup(H)^{\delta\cdot\ctw(\tG)}\cdot |V(\tG)|^{\Oh(1)}\leq \mimsup(H)^{\delta\cdot(k+g(H))}\cdot |V(G)|^{\Oh(1)}=\mimsup(H)^{\delta\cdot k}\cdot |V(G)|^{\Oh(1)},
   \]
   which by \cref{thm:list-gen-ord}~(1.) contradicts the ETH.

Finally, let $\eps>0$ and suppose there is an algorithm $\cA_2$ that solves every instance $G'$ of $\homo{H}$ in time $(\mimsup(H)-\eps)^{\ctw(G')}\cdot |V(G')|^{\Oh(1)}$. Then, we can call the above reduction for any instance $(G,L)$ of $\lhomo{H}$ given along with a linear ordering of $V(G)$ of width $k$ and use $\cA_2$ to solve it in time:
\begin{align*}
  & (\mimsup(H)-\eps)^{\ctw(\tG)}\cdot |V(\tG)|^{\Oh(1)}\leq (\mimsup(H)-\eps)^{k+g(H)}\cdot |V(G)|^{\Oh(1)} \\ & =  (\mimsup(H)-\eps)^{k}\cdot |V(G)|^{\Oh(1)},
\end{align*}
   which, by \cref{thm:list-gen-ord}~(2.), contradicts the SETH.
\end{proof}

\subsection{Proof of \cref{thm:list-bip-ord}}
We are left with proving \cref{thm:list-bip-ord}.
Let us start with introducing some tools developed in the literature~\cite{DBLP:conf/esa/OkrasaPR20,DBLP:conf/stacs/PiecykR21}.

\paragraph{Decomposable graphs.}
Okrasa et al.~\cite{DBLP:conf/esa/OkrasaPR20} defined the following decomposition of biparitte graphs which turned out to be useful in solving \lhomo{H}.

\begin{definition}[Bipartite decomposition]\label{def:bipartite-decomposition}
Let $H$ be a bipartite graph with bipartition classes $X,Y$.
A partition of $V(H)$ into an ordered triple of sets $(D,N,R)$ is a \emph{bipartite decomposition} if the following conditions are satisfied (see \cref{fig:decomp})
\begin{enumerate}
\item $N$ is non-empty and separates $D$ and $R$, \label{it:bipdecomp-separator}
\item $|D\cap X| \geq 2$ or $|D \cap Y| \geq 2$, \label{it:bipdecomp-geq2}
\item $N$ induces a biclique in $H$, \label{it:bipdecomp-biclique}
\item $(D \cap X) \cup (N \cap Y)$ and $(D \cap Y) \cup (N \cap X)$ induce bicliques in $H$. \label{it:bipdecomp-complete}
\end{enumerate}
If $H$ admits a bipartite decomposition, we call it \emph{decomposable}, otherwise we call it \emph{undecomposable}.
\end{definition}

\begin{figure}[h!]
\centering{\begin{tikzpicture}[every node/.style={draw,circle,fill=white,inner sep=0pt,minimum size=30pt},every loop/.style={}]
\node (dx) at (0,0) {};
\node (dy) at (2,0) {};
\node (nx) at (0,-1.5) {};
\node (ny) at (2,-1.5) {};
\node (rx) at (0,-3) {};
\node (ry) at (2,-3) {};
\draw[very thick] (dx)--(ny)--(nx)--(dy);
\draw[color=orange] (nx)--(ry)--(rx)--(ny);
\draw[color=orange] (dx)--(dy);
\draw[dashed] (-1,0.7)--++(4,0)--++(0,-1.4)--++(-4,0)--++(0,1.4);
\draw[dashed] (-1,-0.8)--++(4,0)--++(0,-1.4)--++(-4,0)--++(0,1.4);
\draw[dashed] (-1,-2.3)--++(4,0)--++(0,-1.4)--++(-4,0)--++(0,1.4);
\node[draw=none, fill=none, label=left:{$D$}] (d) at (-1,0) {};
\node[draw=none, fill=none, label=left:{$N$}] (n) at (-1,-1.5) {};
\node[draw=none, fill=none, label=left:{$R$}] (r) at (-1,-3) {};
\end{tikzpicture}}
\caption{Bipartite decomposition $(D,N,R)$. Circles denote independent sets. A black line denotes that there are all possible edges between sets, an orange one that there might be some edges, and the lack of a line denotes that there are no edges between sets. The figure is copied from~\cite{DBLP:conf/esa/OkrasaPR20} by courtesy of the authors.}\label{fig:decomp}
\end{figure}
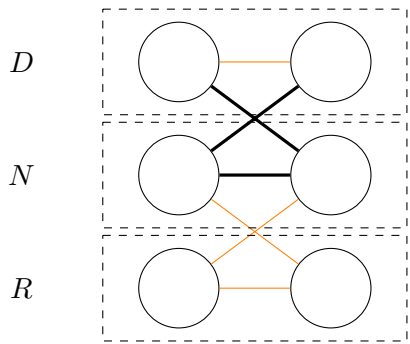

We observe that the property of being incomparable is actually stronger that being undecomposable.

\begin{proposition}\label{prop:undecomposable}
Every connected bipartite incomparable graph is undecomposable.
\end{proposition}
\begin{proof}
Let $H$ be a connected bipartite incomparable graph with bipartition classes $X,Y$, and, for contradiction, suppose that it admits a decomposition $(D,N,R)$.

Note that the neighborhood of every vertex in $D \cap X$ (resp. $D \cap Y)$ is contained in the neighborhood of any vertex in $N \cap X$ (resp. $N \cap Y$).
Thus we have $D \cap X = \emptyset$ or $N \cap X = \emptyset$, and $D \cap Y = \emptyset$ or $N \cap Y = \emptyset$.

Suppose that $D \cap X = \emptyset$ (the case that $D \cap Y = \emptyset$ is symmetric). This implies that $|D \cap Y| \geq 2$, and all vertices in $D \cap Y$ have the same neighborhood (i.e., $N \cap Y$).
So we conclude that $N \cap X = N \cap Y = \emptyset$, a contradiction with $N$ being non-empty.
\end{proof}

\subsubsection{Basic gadgets}\label{sec:tools}
In this section we introduce gadgets which will be basic building blocks in our reduction.

Throughout this section we assume that $H$ is a bipartite graph, whose complement is not a circular-arc graph and $(\alpha,\beta,\gamma)$ is a fixed triple of vertices from one bipartition class of $H$. We will need the following two gadgets, introduced in \cite{DBLP:conf/stacs/PiecykR21}.

\begin{restatable}[Assignment gadget]{definition}{defAssign}\label{def:assign-gadget}
Let $S$ be an incomparable one-sided set in $H$ and let $v \in S$. An \emph{assignment gadget for $(\alpha,\beta,\gamma)$} is a graph $A_v$ with $H$-lists $L$ and with interface vertices $x,y$, such that:
\begin{enumerate}[({A}1.)]
\item $L(x)=S$ and $L(y)=\{\alpha,\beta,\gamma\}$,

\item for every $u \in S$ and for every $a \in \{\alpha,\beta\}$ there exists a list homomorphism $\vphi: (A_v,L) \to H$ such that $\vphi(x)=u$ and $\vphi(y)=a$,

\item there exists a list homomorphism $\vphi: (A_v,L) \to H$ such that $\vphi(x)=v$ and $\vphi(y)=\gamma$,

\item for every list homomorphism $\vphi: (A_v,L) \to H$ it holds that if $\vphi(y)=\gamma$, then $\vphi(x)=v$.
\end{enumerate}
\end{restatable}

The second gadget is called a \emph{switching gadget}. It is a path $T$ with a special internal vertex $q$, whose list is $\{\alpha,\beta,\gamma\}$, and endvertices with the same list $\{\alpha,\beta\}$.
Mapping both endvertices of $T$ to the same vertex, i.e., mapping both to $\alpha$ or both to $\beta$, allows us to map $q$ to one of $\alpha,\beta$, but ``switching sides'' from $\alpha$ to $\beta$ forces mapping $q$ to $\gamma$.

\begin{restatable}[Switching gadget]{definition}{defSwitch}\label{def:switch-gadget}
A \emph{switching gadget for $(\alpha,\beta,\gamma)$} is a path $T$ of even length with $H$-lists $L$, endvertices $p,r$, called respectively the \emph{input} and the \emph{output} vertex, and one special internal vertex $q$, called the \emph{$q$-vertex}, in the same bipartition class as $p,r$, such that:
\begin{enumerate}[(S1.)]
\item $L(p)=L(r)=\{\alpha,\beta\}$ and $L(q)=\{\alpha,\beta,\gamma\}$,
\item for every $a \in \{\alpha,\beta\}$ there exists a list homomorphism $\vphi: (T,L) \to H$, such that $\vphi(p)=\vphi(r)=a$ and $\vphi(q) \neq \gamma$,
\item there exists a list homomorphism $\vphi: (T,L) \to H$, such that $\vphi(p)=\alpha$, $\vphi(r)=\beta$, and $\vphi(q)=\gamma$,
\item for every list homomorphism $\vphi: (T,L) \to H$, if $\vphi(p)=\alpha$ and $\vphi(r)=\beta$, then $\vphi(q)=\gamma$.
\end{enumerate}
\end{restatable}
Note that in a switching gadget we do not care about homomorphisms that map $p$ to $\beta$ and $r$ to $\alpha$.

Later, when discussing assignment and switching gadgets, we will use the notions of $x$-, $y$-, $p$-, $q$-, and $r$-\emph{vertices} to refer to the appropriate vertices introduced in the definitions of the gadgets.

The third gadget that we will use in the reduction is a \emph{variable gadget}.

\begin{definition}[Variable gadget]\label{def:var-gadget}
Let $k\in \N$ and let $S\subseteq S_1\times \ldots \times S_k$, where $S_1,\ldots,S_k\subseteq V(H)$. A \emph{variable gadget for $S$} is a graph $X$ with $H$-lists $L$ and $k$ interface vertices $x_1,\ldots,x_k$ such that:
\begin{enumerate}[(V1.)]
\item for every list homomorphism $\vphi: (X,L)\to H$ it holds that $(\vphi(x_1),\ldots,\vphi(x_k))\in S$,
\item for every $(s_1,\ldots,s_k)\in S$, there exists a list homomorphism $\vphi: (X,L)\to H$ such that $(\vphi(x_1),\ldots,\vphi(x_k))=(s_1,\ldots,s_k)$.
\end{enumerate}
\end{definition}

We will construct a variable gadget from the following two types of gadgets, first being an \emph{indicator gadget}.

\begin{definition}[Indicator gadget]
    Let $S$ be an incomparable one-sided set of $H$ and let $s\in S$. \emph{An indicator gadget $I(s)$ for $(\alpha,\beta)$}  is a graph $I$ with $H$-lists $L$ and with interface vertices $x,y$, input and output, such that:
    \begin{enumerate}[({I}1.)]
        \item $L(x)=S$ and $L(y)=\{\alpha,\beta\}$.
        \item For every list homomorphism $\vphi: (I,L)\to H$, if $\vphi(x)=s$, then $\vphi(y)=\beta$.
        \item For every $s'\in S\setminus\{s\}$, there exists a list homomorphism $\vphi: (I,L)\to H$ such that $\vphi(x)=s'$ and $\vphi(y)=\alpha$.
        \item There exists a list homomorphism $\vphi: (I,L)\to H$ such that $\vphi(x)=s$ and $\vphi(y)=\beta$.
    \end{enumerate}
    
\end{definition}

The second building block of the variable gadget is the following.

\begin{definition}[$\nand_k$-gadget]
    \emph{A $\nand_k$-gadget for $(\alpha,\beta)$} is a graph $F$ with $H$-lists $L$ and with $k$ interface vertices $a_1,\ldots,a_k$ such that:
    \begin{enumerate}
        \item For every $i\in [k]$, we have $L(a_i)=\{\alpha,\beta\}$.
        \item For every list homomorphism $\vphi: (F,L)\to H$ we have $(\vphi(a_1),\ldots,\vphi(a_k))\in \{\alpha,\beta\}^k\setminus \{(\beta,\ldots,\beta)\}$. 
        \item For every tuple $(t_1,\ldots,t_k)\in \{\alpha,\beta\}^k\setminus \{(\beta,\ldots,\beta)\}$ there exists a list homomorphism $\vphi: (F,L)\to H$ such that $\vphi(a_i)=t_i$ for every $i\in [k]$.
    \end{enumerate}
\end{definition}

By the following lemma from \cite{DBLP:conf/esa/OkrasaPR20}, we can construct the two latter gadgets.

\begin{lemma}[Okrasa, Piecyk, Rzążewski~\cite{DBLP:conf/esa/OkrasaPR20}]\label{lem:var-nand}
    Let $H$ be an undecomposable, connected, bipartite graph, whose complement is not a circular-arc graph. Let $S$ be an incomparable set contained in one bipartition class of $H$. There exists a pair $(\alpha,\beta)$ of incomparable vertices from one bipartition class such that there exists a $\nand_k$-gadget for $(\alpha,\beta)$ and for every $s\in S$ there exists an indicator gadget $I(s)$ for $(\alpha,\beta)$.
\end{lemma}


The switching gadget and the assignment gadget can be constructed by the following lemma of~\cite{DBLP:conf/stacs/PiecykR21}.

\begin{lemma}[Piecyk, Rzążewski~\cite{DBLP:conf/stacs/PiecykR21}]\label{lem:gadgets}
 Let $H$ be an undecomposable, connected, bipartite graph, whose complement is not a circular-arc graph. For $\alpha,\beta$ given by \cref{lem:var-nand}, there exists $\gamma$ in the same bipartition  class of $H$, such that the following holds.
\begin{enumerate}
    \item There exist $\alpha'\in N(\alpha)\setminus N(\beta)$ and $\beta'\in N(\beta)\setminus N(\alpha)$.
    \item There exists a switching gadget $T$ for $(\alpha,\beta,\gamma)$.
    \item Let $S$ be an incomparable one-sided set in $H$, such that $|S| \geq 2$. Then for every $v \in S$, there exists an assignment gadget $A_v$ for $(\alpha,\beta,\gamma)$.
\end{enumerate}
\end{lemma}

In the following lemma we will show that we can also construct a variable gadget.

\begin{lemma}\label{lem:var-gadget}
Let $H$ be an undecomposable, connected, bipartite graph, whose complement is not a circular-arc graph. Let $S_1,\ldots,S_k$ be incomparable one-sided sets of $H$, and let $S\subseteq S_1\times\ldots\times S_k$. Then there exists a variable gadget $X$ for $S$.
\end{lemma}

\begin{proof}
    Let $\alpha,\beta$ be the pair given by \cref{lem:var-nand}. First, let us introduce the interface vertices $x_1,\ldots,x_k$ with lists $L(x_i)=S_i$. For every tuple $(s_1,\ldots,s_k)\in S_1\times\ldots\times S_k\setminus S$ we introduce the following gadgets. For $i\in [k]$, we introduce an indicator gadget $I(s_i)$ for $(\alpha,\beta)$ and identify its input vertex with $x_i$. Then we introduce a $\nand_k$-gadget for $(\alpha,\beta)$ and identify its interface vertices with the output vertices of indicator gadgets. This completes the construction of the variable gadget $X$.

    Let us verify that the constructed graph is indeed a variable gadget. First let $(s_1,\ldots,s_k)\in S_1\times\ldots\times S_k\setminus S$ and suppose that there is a list homomorphism $\vphi: (X,L)\to H$ such that $\vphi(x_i)=s_i$ for every $i\in [k]$. Then for every indicator gadget $I(s_i)$ introduced for tuple $(s_1,\ldots,s_k)$, its output vertex must be mapped to $\beta$. These output vertices were identified with interface vertices of a $\nand_k$-gadget, and thus cannot all be mapped to $\beta$, which is a contradiction.

    So now let $(s_1,\ldots,s_k)\in S$. We set $\vphi(x_i)=s_i$ for every $i\in [k]$. It remains to show that $\vphi$ can be extended to the remaining vertices of $X$, i.e., vertices of gadgets introduced for some tuples $(s'_1,\ldots,s'_k)$. Consider such a tuple $(s'_1,\ldots,s'_k)$, i.e., any tuple from $S_1\times\ldots\times S_k\setminus S$. There must be $i\in [k]$ such that $s_i\neq s'_i$. Therefore the output vertex of the indicator gadget $I(S_i)$ can be mapped to $\alpha$. We can extend $\vphi$ to the vertices of the remaining indicator gadgets so that their output vertices are mapped either to $\alpha$ or $\beta$. Since at least one of the interface vertices is not mapped to $\beta$, we can extend $\vphi$ to the remaining vertices of the $\nand_k$-gadget. This completes the proof.
\end{proof}

\subsubsection{Reduction}

In this section we will use the introduced gadgets to reduce the \emph{Constraint Satisfaction Problem} to $\lhomo{H}$.
In the \textsc{CSP}$(q,r)$ problem we are given a set $D$ (domain) of size $q$, and in the input we are given a set of variables $V$ and a set of $r$-ary constraints $C$, each of type $R(v_1,\ldots,v_r)$, where $R\subseteq D^r$ and $(v_1,\ldots,v_r)\in V^r$. The task is to determine whether there exists an assignment $f: V\to D$ such that for every $R(v_1\ldots,v_r)\in C$, we have $(f(v_1),\ldots,f(v_r))\in R$.

The construction in the following lemma is a refinement of the construction from~\cite{DBLP:conf/stacs/PiecykR21}.

\begin{lemma}\label{lem:reduction}
    Let $q,r,k\in \N$ and let $H$ be a connected bipartite incomparable graph, whose complement is not a circular-arc graph and such that $\mim(H^{\otimes k})\geq q$. Let $\phi$ be an instance of $\textsc{CSP}(q,r)$ with $n$ variables and $m$ constraints. In time polynomial in $(n+m)$ we can construct a graph $G$ with $H$-lists $L$ and with a linear ordering $\sigma$ of $V(G)$, such that:
    \begin{enumerate}[(1.)]
        \item $\phi$ is a yes-instance if and only if $(G,L)\to H$,
        \item the width of $\sigma$ is at most $k\cdot n + r\cdot g(k,H)$, where $g$ is some function depending only on $H$ and $k$,
        \item $|V(G)|=(n+m)^{\Oh(1)}$.
    \end{enumerate}
\end{lemma}

\begin{proof}
Recall that by \cref{prop:undecomposable} $H$ is undecomposable, so it satisfies the assumptions of all the lemmas in \cref{sec:tools}.
Let $(\alpha,\beta,\gamma)$ be the triple given by \cref{lem:gadgets} and let $D$ be the domain of $\phi$. We construct $(G,L)$ as follows.

\begin{figure}
\centering{\begin{tikzpicture}[every node/.style={draw,circle,fill=white,inner sep=0pt,minimum size=8pt},every loop/.style={}]
\node[label=left:\footnotesize{$x_C$}] (xc) at (0,0) {};
\foreach \k in {1,3,5,8,10,12}
{
\node (a\k) at (\k,0) {};
}
\node[label=right:\footnotesize{$y_C$}] (yc) at (13,0) {};
\draw (xc)--(a1);
\draw (yc)--(a12);

\foreach \k in {1,3,8,10}
{
\draw (a\k)--++(0.3,0)--++(0,0.3)--++(1.4,0)--++(0,-0.6)--++(-1.4,0)--++(0,0.3);
}
\foreach \i in {3,5,10,12}
{
\draw (a\i)--++(-0.3,0);
}
\foreach \j in {2,4,9,11}
{
\node (q\j) at (\j,0.5) {};
\draw (q\j)--++(0,-0.2);
}
\foreach \k in {0,1,2}
{
\node (b\k) at (0.5+5*\k,3.5) {};
\draw (b\k)--++(0,0.5);
}

\foreach \k in {0,1,2}
{
\node (c\k) at (2.5+5*\k,3.5) {};
\draw (c\k)--++(0,0.5);
}
\foreach \k in {0,1,2}
{
\node (d\k) at (1+5*\k,3.5) {};
\draw (d\k)--++(0,0.5);
\draw[fill=black] (1.5+5*\k,3.5) circle (0.02);
\draw[fill=black] (1.75+5*\k,3.5) circle (0.02);
\draw[fill=black] (2+5*\k,3.5) circle (0.02);
\draw[fill=black] (1.5+5*\k,1.8) circle (0.02);
\draw[fill=black] (1.75+5*\k,1.8) circle (0.02);
\draw[fill=black] (2+5*\k,1.8) circle (0.02);
\draw (d\k)--++(0,-1)--++(0.2,0)--++(0,-1.3)--++(-0.4,0)--++(0,1.3)--++(0.2,0);
\draw (c\k)--++(0,-1)--++(0.2,0)--++(0,-1.3)--++(-0.4,0)--++(0,1.3)--++(0.2,0);
\draw (b\k)--++(0,-1)--++(0.2,0)--++(0,-1.3)--++(-0.4,0)--++(0,1.3)--++(0.2,0);
}

\node (b3) at (3,3.5) {};
\node (c3) at (5,3.5) {};
\node (d3) at (3.5,3.5) {};

\draw (b3)--++(0,0.5);
\draw (c3)--++(0,0.5);
\draw (d3)--++(0,0.5);

\draw (0.35,4)--++(2.3,0)--++(0,0.7)--++(-2.3,0)--(0.35,4);
\draw (2.85,4)--++(2.3,0)--++(0,0.7)--++(-2.3,0)--(2.85,4);
\draw (5.35,4)--++(2.3,0)--++(0,0.7)--++(-2.3,0)--(5.35,4);
\draw (10.35,4)--++(2.3,0)--++(0,0.7)--++(-2.3,0)--(10.35,4);

\node[draw=none, fill=none] (v) at (1.5,4.35) {$X_i$};
\node[draw=none, fill=none] (v) at (4,4.35) {$X'_i$};
\node[draw=none, fill=none] (v) at (6.5,4.35) {$X_i$};
\node[draw=none, fill=none] (v) at (11.5,4.35) {$X_{i+1}$};

\draw (q2)--(0.5,1.2);
\draw (q2)--(1,1.2);
\draw (q2)--(2.5,1.2);
\draw (q4)--(5.5,1.2);
\draw (q4)--(6,1.2);
\draw (q4)--(7.5,1.2);
\draw (q4)--(10.5,1.2);
\draw (q4)--(11,1.2);
\draw (q4)--(12.5,1.2);

\draw (b0) to [bend right] (b3);
\draw (c0) to [bend right] (c3);
\draw (d0) to [bend right] (d3);
\draw (b3) to [bend right] (b1);
\draw (c3) to [bend right] (c1);
\draw (d3) to [bend right] (d1);

\foreach \i in {4,4.3,4.6,8.8,9,9.2}
{
\draw[fill=black] (\i,3.5) circle (0.02);
}
\draw (a5)--++(0.5,0);
\draw (a8)--++(-0.5,0);
\draw[fill=black] (6.2,0) circle (0.02);
\draw[fill=black] (6.5,0) circle (0.02);
\draw[fill=black] (6.7,0) circle (0.02);
\node[draw=none,fill=none,label=below:\footnotesize{$T_C^{f_1}$}] (t) at (2,0.5) {};
\node[draw=none,fill=none,label=below:\footnotesize{$T_C^{f_2}$}] (t) at (4,0.5) {};
\node[draw=none,fill=none,label=below:\footnotesize{$T_C^{f_j}$}] (t) at (9,0.5) {};
\node[draw=none,fill=none,label=below:\footnotesize{$T_C^{f_\ell}$}] (t) at (11,0.5) {};

\end{tikzpicture}}
\caption{The path $P_C$ for a constraint $C$ and variable gadgets connected to $P_C$ with assignment gadgets.}\label{fig:reduction}
\end{figure}
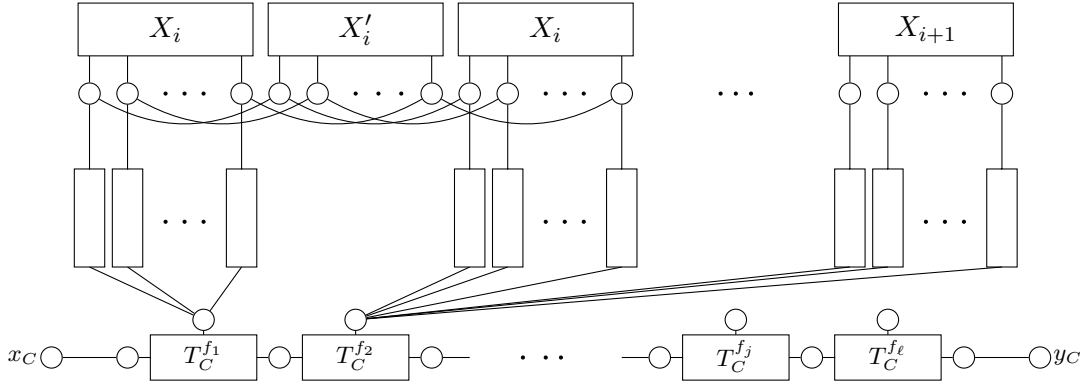

\paragraph{Constraint gadgets.} Let $C_1,\ldots,C_m$ be the constraints of $\phi$. For each constraint $C$, we construct a \emph{constraint gadget}, which is a path $P_{C}$ with $H$-lists $L$, as follows. We start with vertices $x_C$ and $y_C$ with lists $L(x_C)=\{\alpha'\}$ and $L(y_C)=\{\beta'\}$, where $\alpha',\beta'$ are as in \cref{lem:gadgets}. Then for every assignment $f$ of the variables in $C$ that satisfies $C$ we introduce a switching gadget $T_C^f$ with $q$-vertex $q_C^f$. We connect the introduced gadgets by identifying the $r$-vertex of the preceding switching gadget with the $p$-vertex of the following one. Moreover, we add edges from the $p$-vertex of the first switching gadget to $x_C$ and from the $r$-vertex from the last switching gadget to $y_C$. This completes the construction of constraint gadgets.

\paragraph{Variable gadgets.} Let $v_1,\ldots, v_n$ be the variables of $\phi$. Let $S,S'\subseteq V(H^{\otimes k})$ be the sets of endpoints of the induced matching $M$ of size $q$ in $H^{\otimes k}$ and such that $S\subseteq X^k$, $S'\subseteq Y^k$ for $X,Y$ being the bipartition classes of $H$. By $S_i$ (resp. $S'_i$) we denote the projection of $S$ (resp. $S'$) on $i$th coordinate. Note that each $S_i$ and each $S'_i$ is incomparable (since $H$ is incomparable) and one-sided. For each variable $v_i$, for each constraint $C$ that contains $v_i$, and for each satisfying assignment $f$ of the variables of $C$ we introduce the variable gadgets $X_{i,f,C}$ and $X_{i,f,C}'$ respectively for the sets $S$ and $S'$, and with interface vertices respectively $x_{i,f,C}^1,\ldots,x_{i,f,C}^k$ and $x'^1_{i,f,C},\ldots,x'^k_{i,f,C}$.

\paragraph{Connecting variables and constraints.} Note that since $|S|=q$, there is a bijection between $S$ and possible assignments of one variable of $\phi$, let us fix one, say $\mu: S \to D$. Let $X_{i,f,C}$ be the variable gadget introduced for the variable $v_i$ and a satisfying assignment $f$ for a constraint $C$. Let $(s_1,\ldots,s_k)$ be the tuple from $S$ that corresponds to $f(v_i)$, i.e., $\mu((s_1,\ldots,s_k))=f(v_i)$. For each interface vertex $x_{i,f,C}^j$ we introduce an assignment gadget $A_{s_j}$ and identify its $x$-vertex with $x_{i,f,C}^j$ and its $y$-vertex with $q_C^f$. Then we add the edges $x_{i,f,C}^jx'^j_{i,f,C}$. Finally, let us fix an ordering of $q$-vertices so that $q_1$ precedes a vertex $q_2$ if:
\begin{itemize}
\item $q_1$ belongs to path $P_{C_i}$ and $q_2$ belongs to $P_{C_j}$, and $i<j$, or
\item $q_1$ and $q_2$ belong to the same path $P_C$ and $q_1$ precedes $q_2$ on $P_C$.
\end{itemize}
Let $q_C^f$ and $q_{C'}^{f'}$ be consecutive $q$-vertices. We add all edges $x'^j_{i,f,C}x^j_{i,f',C'}$. This completes the construction of $(G,L)$ (see \cref{fig:reduction}).

\paragraph{Correctness.} Let us verify that $(G,L)\to H$ if and only if $\phi$ is satisfiable.
\begin{claim}
If $(G,L)\to H$, then $\phi$ is satisfiable.
\end{claim}

\begin{claimproof}
Let $\vphi$ be a list homomorphism from $(G,L)$ to $H$. First observe that since $M$ is an induced matching in $H^{\otimes k}$, all variable gadgets $X_{i,f',C}$ corresponding to the variable $v_i$ have its interface vertices mapped to the same tuple $(s_1,\ldots,s_k)$ of $S$. Indeed, since for every variable gadget $X_{i,f',C}$ its interface vertices are adjacent to corresponding interface vertices of the gadget $X'_{i,f',C}$, mapping interface vertices of $X_{i,f',C}$ to $(s_1,\ldots,s_k)$ forces mapping the interface vertices of $X'_{i,f',C}$ to a tuple $(s'_1,\ldots,s'_k)$ such that $s_\ell s'_\ell\in E(H)$ for every $\ell\in [k]$, and thus $(s_1,\ldots,s_k)$ is adjacent to $(s'_1,\ldots,s'_k)$ in $H^{\otimes k}$. Since the interface vertices can be mapped only to tuples that are endpoints of the induced matching $M$ in $H^{\otimes k}$, the tuple $(s_1,\ldots,s_k)$ uniquely determines $(s_1',\ldots,s'_k)$. Similarly, any mapping of the interface vertices of $X'_{i,f',C}$ uniquely determines the mapping of the interface vertices of the variable gadget $X_{i',f'',C'}$ following $X'_{i,f',C}$.
Therefore we can set $f(v_i)=\mu((s_1,\ldots,s_k))$, for $(s_1,\ldots,s_k)$ being the tuple of vertices that the interface vertices of any variable gadget $X_{i,f',C}$ are mapped to.

Since the vertices $\alpha,\alpha',\beta,\beta'$ induce a matching in $H$, for every constraint $C$, the $p$-vertex of the first switching gadget on $P_C$ must be mapped to $\alpha$ and the $r$-vertex of the last switching gadget on $P_C$ must be mapped to $\beta$. Therefore, there must be a switching gadget $T_C^{f'}$ such that its $p$-vertex is mapped to $\alpha$ and its $r$-vertex is mapped to $\beta$. By property (S4.), the vertex $q_C^{f'}$ must be mapped to $\gamma$. Recall that for every variable $v_i$ that appears in $C$, we introduced $k$ assignment gadgets $A_{s_j}$, where $(s_1,\ldots,s_k)$ is the tuple such that $f'(v_i)=\mu((s_1,\ldots,s_k))$. By the property (A4.) of the assignment gadget, each vertex $x_{i,f',C}^j$ is mapped by $\vphi$ to $s_j$. Thus $f=f'$ on the variables from $C$, so $f$ satisfies $C$. This completes the proof of the claim.
\end{claimproof}

\begin{claim}
If $\phi$ is satisfiable, then $(G,L)\to H$.
\end{claim}

\begin{claimproof}
Let $f$ be a satisfying assignment of the variables in $\phi$. We construct $\vphi:(G,L)\to H$ as follows. For every variable gadget $X_{i,f'}$ we set its interface vertices $x_{i,f,C}^j$ to $s_j$, where $(s_1,\ldots, s_k)=\mu^{-1}(f(v_i))$ and we map the vertices $x'^j_{i,f,C}$ to $s_j$, where $(s'_1,\ldots, s'_k)$ is the private neighbor of $(s_1,\ldots,s_k)$ in $H^{\otimes k}$ in the set $S'$. By \cref{def:var-gadget}, this mapping can be extended to every variable gadget. Now consider a constraint $C$ and a switching gadget $T_C^f$. For this gadget, we map its $p$-vertex to $\alpha$, its $q$-vertex to $\gamma$, its $r$-vertex to $\beta$, and we extend $\vphi$ to all other vertices of this gadget by (S3.). For all other switching gadgets on $P_C$ that precede $T_C^f$, we map their $p$-vertices and $r$-vertices to $\alpha$ and we extend $\vphi$ to every switching gadget so that no $q$-vertex is mapped to $\gamma$, which can be done by (S2.). Similarly, for all switching gadgets that follow $T_C^f$ on $P_C$ we map their $p$-vertices and $r$-vertices to $\beta$ and extend $\vphi$ on remaining vertices of the gadgets so that no $q$-vertex is mapped to $\gamma$. It remains to map the vertices of the assignment gadgets. Let $A$ be an assignment gadget that joins an interface vertex $x_{i,f',C}^j$ with a $q$-vertex $q_C^{f'}$. If $f'=f$, then $\vphi(q_C^f)=\gamma$ and $\vphi(x_{i,f}^j)=s_j$, where $(s_1,\ldots,s_k)=\mu^{-1}(f(v_j))$. Recall that $A$ must be an assignment gadget $A_{s_j}$. Therefore, by (A3.) the homomorphism $\vphi$ can be extended to remaining vertices of the gadget. Finally, if $f'\neq f$, then $\vphi(q_C^{f'})\in\{\alpha,\beta\}$. By (A2.) the homomorphism $\vphi$ can be extended to remaining vertices of the gadget. This completes the proof of the claim.
\end{claimproof}

\paragraph{Cutwidth.} Now we will construct a linear ordering $\sigma$ of $V(G)$ of width at most $k\cdot n + r \cdot g(k,H)$, where $g$ is a function that depends only on $k$ and $H$. First we order vertices of the constraint gadgets so that we first put the vertices from $P_{C_1}$ in the order from the path, then from $P_{C_2}$, and so on. Then we modify this ordering by inserting, immediately after each $q$-vertex $q_C^f$, the vertices from the gadgets $X_{i,f,C}$ and $X'_{i,f,C}$ such that $v_i$ is a variable that appears in $C$, and we place there also the vertices of the assignment gadgets that connect $X_{i,f,C}$ with $q_C^f$. The ordering of the vertices of these gadgets is arbitrary. This completes the construction of the ordering of $V(G)$. Now let us verify that it has desired width. Consider any cut. The edges that can cross this cut are:
\begin{itemize}
\item at most one edge from a constraint gadget,
\item at most $r\cdot g(k,H)$ edges of variable gadgets and assignment gadgets corresponding to the same $q$-vertex, where $g$ depends only on $k$ and $H$,
\item at most $n\cdot k$ edges that connect consecutive variable gadgets.
\end{itemize}

Therefore the cutwidth of $G$ is at most $n\cdot k +r\cdot g(k,H)$.



Finally, observe that the construction of $(G,L)$ is performed in time polynomial in $(n+m)$ and thus $|V(G)|=(n+m)^{\Oh(1)}$. This completes the proof.
\end{proof}

\subsubsection{Wrapping up the proof}

We will use \cref{lem:reduction} to provide a series of lower bounds for the complexity of $\lhomo{H}$. We will use the folowing result of Lampis~\cite{DBLP:journals/siamdm/Lampis20}. The first part of the theorem was not stated in the paper, but it follows from the proof of the second part.

\begin{theorem}[Lampis~\cite{DBLP:journals/siamdm/Lampis20}]\label{thm:csp}

\begin{enumerate}[(1.)]
    \item Assuming the ETH, there exists $\delta>0$ such that for every $q\geq 2$ the following holds. The \textsc{CSP}$(q,3)$ problem on $n$ variables and $m$ clauses cannot be solved in time $q^{\delta \cdot n}\cdot (n+m)^{\Oh(1)}$.
    \item Assuming the SETH, for every $q\geq 2$ and $\eps>0$ there exists $r$ such that the following holds. The \textsc{CSP}$(q,r)$ problem on $n$ variables and $m$ clauses cannot be solved in time $(q-\eps)^n\cdot(n+m)^{\Oh(1)}$.
\end{enumerate}

\end{theorem}

 By combining \cref{thm:csp} with \cref{lem:reduction} we obtain the following.

\begin{theorem}\label{thm:reduction}
Let $\cH_{0}$ denote the set of connected incomparable bipartite graphs whose complement is not a circular-arc graph.
\begin{enumerate}[(1.)]
    \item Assuming the ETH, there exists $\delta>0$ such that for every $H\in \cH_0$ and for every $k\in \N$, the following holds. There is no algorithm solving every instance $(G,L)$ of $\lhomo{H}$ given with a linear ordering of width $t$ in time $\mim(H^{\otimes k})^{\frac{1}{k}\cdot \delta\cdot t}\cdot |V(G)|^{\Oh(1)}$.
    \item Assuming the SETH, for every $H\in \cH_0$, for every $k\in \N$, and for every $\eps>0$, the following holds. There is no algorithm solving every instance $(G,L)$ of $\lhomo{H}$ given with a linear ordering of width $t$ in time $(\mim(H^{\otimes k})^{\frac{1}{k}}-\eps)^{t}\cdot |V(G)|^{\Oh(1)}$.
\end{enumerate}
 
\end{theorem}

\begin{proof}[Proof of (1.)]
    First assume the ETH and let $\delta>0$ be given by \cref{thm:csp} (1.). Furthermore let $H\in \cH_0$ and let $q=\mim(H^{\otimes k})$. Suppose that there is an algorithm $\cA_1$ that solves every instance $(G,L)$ of $\lhomo{H}$ given with a linear ordering of width $t$ in time $\mim(H^{\otimes k})^{\frac{1}{k}\cdot \delta\cdot t}\cdot |V(G)|^{\Oh(1)}$.  Let $\phi$ be an instance of \textsc{CSP}$(q,3)$ on $n$ variables and $m$ clauses. We call \cref{lem:reduction} to construct an instance $(G,L)$ of $\lhomo{H}$ with a linear ordering $\sigma$ of width $t$ satisfying the conditions given in \cref{lem:reduction}. Thus we can solve the instance $\phi$ by calling $\cA_1$ in time:
    \begin{align*}
      \mim(H^{\otimes k})^{\frac{1}{k}\cdot \delta\cdot t}\cdot |V(G)|^{\Oh(1)}\leq q^{\frac{1}{k}\cdot \delta\cdot (k\cdot n + r\cdot g(k,H))}\cdot |V(G)|^{\Oh(1)}=q^{\delta\cdot n}\cdot |V(G)|^{\Oh(1)},
    \end{align*}
    which is a contradiction by \cref{thm:csp}.
\end{proof}

\begin{proof}[Proof of (2.)] Now assume the SETH. Let $\eps>0$, let $H\in \cH_0$ and suppose there is an algorithm $\cA_2$ that solves every instance $(G,L)$ of $\lhomo{H}$ given with a linear ordering of width $t$ in time $(\mim(H^{\otimes k})^{\frac{1}{k}}-\eps)^{t}\cdot |V(G)|^{\Oh(1)}$. Let $q=\mim(H^{\otimes k})$, let $\eps'=q-(q^{\frac{1}{k}}-\eps)^k>0$, and let $r=r(q,\eps')$ be given by \cref{thm:csp} (2.). Let $\phi$ be an instance of \textsc{CSP}$(q,r)$ on $n$ variables and $m$ clauses. We call \cref{lem:reduction} to construct an instance $(G,L)$ of $\lhomo{H}$ with a linear ordering $\sigma$ of width $t$ satisfying the conditions given in \cref{lem:reduction}. 

    Since the instance $(G,L)$ is equivalent to the instance $\phi$, we can use $\cA_2$ to solve $\phi$. The reduction is performed in time polynomial in $(n+m)$, and $\cA_2$ runs in time:
    \begin{align*}
     & (\mim(H^{\otimes k})^{\frac{1}{k}}-\eps)^{t}\cdot |V(G)|^{\Oh(1)}\leq (q^{\frac{1}{k}}-\eps)^{k\cdot n + r\cdot g(k,H)}\cdot (n+m)^{\Oh(1)}= (q^{\frac{1}{k}}-\eps)^{k\cdot n}\cdot (n+m)^{\Oh(1)} \\
     & = (q-\eps')^n \cdot (n+m)^{\Oh(1)},
    \end{align*}
    which is a contradiction by \cref{thm:csp}.
\end{proof}

By definition, for a bipartite $H$ we have $\mimsup(H)=\limsup_{k\to\infty}\mim(H^{\otimes k})^{1/k}$, so \cref{thm:reduction} yields the following.


\ListBipOrd*

\begin{proof}[Proof of (1.)]
 Let $H\in \cH_0$ and let $(G,L)$ be an instance of $\lhomo{H}$ given with a linear ordering of width $t$. Let $\delta'>0$ be given by \cref{thm:reduction} and let $0<\delta<\delta'$. Furthermore, let $k\in \N$ be such that $\mim(H^{\otimes k})^{\frac{1}{k}}>\mimsup(H)^{\frac{\delta}{\delta'}}$ -- it exists by the definition of $\mimsup$ and the fact that $\frac{\delta}{\delta'}<1$. If $(G,L)$ can be solved in time

\[
\mimsup(H)^{\delta \cdot t}\cdot |V(G)|^{\Oh(1)} < \mim(H^{\otimes k}) ^{\frac{\delta'}{k} \cdot t}\cdot |V(G)|^{\Oh(1)},
\]
then by \cref{thm:reduction} it contradicts the ETH.
\end{proof}

\begin{proof}[Proof of (2.)]
Let $H\in \cH_0$ and let $(G,L)$ be an instance of $\lhomo{H}$ given with a linear ordering of width $t$. Let $\eps>0$ and let $\eps'=\frac{\eps}{2}$. By the definition of $\mimsup(H)$, there exists $k\in \N$ such that $\mim(H^{\otimes k})^{\frac{1}{k}}-\eps'\geq\mimsup(H)-\eps$. If $(G,L)$ can be solved in time 
\[
(\mimsup(H)-\eps)^{t}\cdot |V(G)|^{\Oh(1)} \leq (\mim(H^{\otimes k})^{\frac{1}{k}}-\eps')^{t}\cdot |V(G)|^{\Oh(1)},
\]
then by \cref{thm:reduction} it contradicts the SETH.
\end{proof}

\section{Relations between various rank parameters}
\label{sec:cis}
In this section, we state various combinatorial results about the parameters studied in this paper.
\subsection{Comparing $\hg$ and $\mimsup$}
\label{subsec:hgvsmimsup}
We first make some simple observations and will then show that (perhaps surprisingly) mimsup can be much larger than him. 
\begin{lemma}
\label{lem:mimsupgeqhg}
Let $A$ be a matrix.
Then $\mimsup(A)\geq \hg(A)\geq \mim(A)$.
\end{lemma}
\begin{proof}
The second inequality follows directly since each induced matching is a half-induced matching. We prove the first inequality.

Let $R=\{a_1,\dots,a_i\}$ and $C=\{b_1,\dots,b_i\}$ be the rows and columns respectively of a maximum half-induced matching in $A$. We may assume that these are ordered such that $A[a_j,b_j]=1$ for all $j\in [i]$ and $A[a_k,b_j]=0$ for all $k< j$. For integers $s\geq 1$, we consider the submatrix of $A^{\otimes is}$ induced on the rows consisting of `balanced' sequences
\[
\{(r_1,\dots,r_{is})\in R^{is} \ | \ |\{\ell:r_\ell=a_j\}|=s, \text{ for every } j\in[i]\}
\]
and similarly for the columns
\[
\{(c_1,\dots,c_{is})\in C^{is} \ | \ |\{\ell :c_\ell=b_j\}|=s, \text{ for every } j\in[i]\}.
\]
We claim this forms an induced matching of size $\binom{is}{s,\dots,s}=i^{(1+o(1))is}$, which shows that $\mimsup(A)\geq i=\hg(A)$. Since the size is clear from the definition, it only remains to check that it indeed forms an induced matching. To show this, we explain why the row of
\[
r=(a_1,\dots,a_1,a_2,\dots,a_2,\dots,a_i,\dots,a_i)
\]
has a single one entry in column
\[
c=(b_1,\dots,b_1,b_2,\dots,b_2,\dots,b_i,\dots,b_i).
\]
The other cases follow by symmetry. It is clear that $A^{\otimes is}[r,c]=1$. Any column $c'$ with $A^{\otimes is}[r,c']=1$ must have $A[a_1,c'_j]=1$ for all $j\in [s]$. But $A[a_1,b_j]=0$ when $j>1$, so this implies that $\{j:c_j'=b_1\}=[s]$. Similarly,
we require $A[a_2,c'_j]=1$ for all $j\in [s+1,2s]$, but $c'$ has already `used' all its $b_1$'s and so $\{j \ | \ c_j'=b_2\}$ needs to be $[s+1,2s]$.
Continuing inductively, we find that $\{j \ | \ c_j'=b_k\}=[(k-1)s+1,ks]$ for all $k\in [i]$, that is, $c'=c$.
\end{proof}
We are now ready to prove previously claimed bounds.
\combhgbnd*
\begin{proof}
Lemma~\ref{lem:mimsupgeqhg} shows the first inequality.
The proof of Lemma~\ref{lem:RepOne} implies that for any matrix $A$, $\mim(A^{\otimes k})\leq k^{\hg(A) k}$, since no smaller representative set can be found when the rows induce an induced matching (with some set of columns).
\end{proof}
The following observation implies that sequences such as $\mim(A^{\otimes 2^k})^{1/2^k}$ are non-decreasing.
\begin{lemma}
\label{lem:mim_multiplicative}
Given two matrices $A$ and $B$,
    $\mim(A\otimes B)\geq \mim(A)\mim(B)$. In particular,  $\mimsup(A^{\otimes k})=\mimsup(A)^k$.
\end{lemma}
\begin{proof}
    Suppose that $(r_1,c_1),\dots,(r_a,c_a)$ forms a maximum induced matching of size $a$ in $A$, in the sense that $A[(r_1,\dots,r_a),(c_1,\dots,c_a)]$ induces the identity matrix. Similarly, suppose the submatrix of $B$ induced on rows $r_1',\dots,r_b'$ and columns $c_1',\dots,c_b'$ (in that order) induces the identity matrix. In $A\otimes B$, we may consider the submatrix induced on $ab$ rows
    \[
    \{(r_i,r_j') \ | \ i\in [a],j\in [b]\}
    \]
    and columns
    \[
    \{(c_i,c_j') \ | \ i\in [a],j\in [b]\}.
    \]
    In this case, $(A\otimes B)[(r_i,r_j'),(c_{v},c_{w}')]=A[r_i,c_v]B[r_j',c_w']$ is 1 if and only if $i= v$ and $j=w$. This gives an induced matching of size $ab=\mim(A)\mim(B)$, as desired.

Since the limit $\lim_{k\to\infty}\mim(A^{\otimes k})^{1/k}$ exists (see Section \ref{sec:prel}), the subsequence $(\mim(A^{\otimes k\ell})^{1/k\ell})_\ell$ has the same limit as $(\mim(A^{\otimes k'})^{1/k'})_{k'}$, that is,
    \[
 \mimsup(A^{\otimes k})^{1/k}=\lim_{\ell\to\infty} \mim(A^{\otimes k\ell })^{1/k\ell}= \lim_{k' \to \infty}\mim(A^{\otimes k'})^{1/k'}=\mimsup(A).
    \]
    This shows the `in particular'.
\end{proof}
The results above also imply that
\[
\mimsup(A)=\limsup_{k\to \infty }\hg(A^{\otimes k})^{1/k}.
\]
In principle, $\mim(A^{\otimes k})^{1/k}$ could grow very slowly and it could take very long to become as large as $\hg(A)$ (and in fact it may stay smaller for all $k\in \N$). Our next result gives some form of guarantee: already for $k=2$, $\mim(A^{\otimes k})^{1/k}\geq \sqrt{\hg(A)}$. 
\begin{lemma}
     $\mim(A^{\otimes 2})\geq \hg(A)$.
\end{lemma}
\begin{proof}
Let $a_1,\dots,a_i$ and $b_1,\dots,b_i$ be the rows and columns respectively of a maximum half-induced matching in $A$. We may assume that these are ordered such that $A[a_j,b_j]=1$ for all $j\in [i]$ and $A[a_k,b_j]=0$ for all $k< j$. In $A^{\otimes 2}$ we claim that the rows $r_1=(a_1,a_i),r_2=(a_2,a_{i-1}),\dots,r_i=(a_i,a_1)$ and the columns $c_1=(b_1,b_{i}),c_2=(b_2,b_{i-1}),\dots,c_i=(b_i,b_1)$ induce a matching of size $i= \hg(A)$. Indeed, for $j,j'\in [i]$, by definition
\[
A[c_j,r_{j'}]=A[a_j,b_{j'}]A[a_{i+1-j},b_{i+1-j'}],
\]
for which both terms are 1 if $j=j'$ and at least one term is 0 if $j\neq j'$.
\end{proof}
We now move to the main result of this section. At first glance, it may be natural to conjecture that in fact $\mimsup(A)=\hg(A)$ for all matrices $A$. This is however not true, as the following result shows.
\begin{theorem}
\label{thm:hgversusmimsup}
There is a constant $C>1$, such that for all sufficiently large integers $h$, there is a symmetric $(2h\times 2h)$ matrix $A$ with $\hg(A)\leq 10\log_2 h$ and $\mimsup(A)\geq \sqrt{h}$.
\end{theorem}
In order to ensure the matrix has no large half-induced matching, we will in fact ensure it has no large blocks of zeros. Moreover, we will show that not just the mimsup is large, we already find a large induced matching in the second Kronecker product. This lemma immediately implies the theorem above,     since if $M$ has a half-induced matching of size $2\ell$, then it has an $\ell\times \ell$ all zeros square submatrix.
\begin{lemma}
\label{thm:hgversusmimsup_stronger}
There is a constant $C>1$, such that for all sufficiently large $h$, there is a symmetric $(2h\times 2h)$ matrix $A$ without $2\lceil 2\log_2h\rceil \times 2\lceil 2\log_2h\rceil$ blocks of zeros yet $\mim(A^{\otimes 2})\geq h$. 
\end{lemma}
\begin{proof}
We write $\cdot$ for the Hadamard (pointwise multiplication) product. Then $A_1\cdot A_2$ is a submatrix of $A^{\otimes 2}$ when $A_1$ and $A_2$ are obtained as a `stretch' of $A$: that is, we are allowed to repeat rows and columns of $A$ as often as we like. 

Let $M$ be a random $(h\times h)$-matrix, with $1$'s on the diagonal and the entries below the diagonal sampled independently uniformly at random from $\{0,1\}$, where the entries above the diagonal are chosen in order to make a symmetric matrix. 
Let $M'$ be the `complement' of $M$: the diagonal entries are still 1 ($M_{ii}'=M_{ii}=1$) but all other entries are flipped ($M_{ij}'=1-M_{ij}$ for $i\neq j$). By construction, $M\cdot M'$ is the $(h\times h)$-identity matrix and $M,M'$ are both symmetric with the same distribution.

We first show that with high probability, $M$ does not contain a $\ell\times \ell$ block of zeros for $\ell=\lceil 2\log_2 h\rceil$. Note that any such block is given by rows $r_1,\dots,r_\ell$ and columns $c_1,\dots,c_\ell$. All must be distinct (since the diagonal entries are non-zero). In particular, all $M[r_i,c_j]$ entries are sampled independently.
By a union bound, the probability that $M$ has an $\ell \times \ell$ block consisting of only zeros is at most
\[
\binom{h}{2\log_2 h}^2 \left(\frac12\right)^{(2\log_2 h)^2}\leq \left(\frac{eh}{2\log_2 h}\right)^{4\log_2 h} \left(\frac1h\right)^{4\log_2 h}=o(1).
\]
(This is the `standard Ramsey calculation'.) So for $h$ sufficiently large, such a block does not exist in $M$ with high probability. 
Since $M$ and $M'$ have the same distribution, the same holds for $M'$. We set \[
A= \begin{pmatrix}
M& M'\\
M' & M\\
\end{pmatrix}.
\]
Then $A$ has no $2\ell  \times2\ell$ block of zeros. At the same time, $A^{\otimes 2}$ contains the Hadamard product $M\cdot M'$ as submatrix. So 
\[
\hg(A\otimes A)\geq \mim(M\cdot M')= h.
\]
We showed that for all $h$ sufficiently large, there is a matrix $A$ without $2\lceil 2\log h\rceil \times 2\lceil 2\log h\rceil$ block of zeros yet 
$\hg(A^{\otimes 2})\geq h$, as desired.
\end{proof}
The proof above gives examples of matrices $M,M'$ for which the Hadamard product $\cdot$ is not multiplicative, that is, $\hg(M\cdot M')>\hg(M)\hg(M')$.
We remark that $\mim(A\cdot B)>\mim(A)\mim(B)$ is also possible, e.g. 
\[
6=\mim(A\cdot B)>2\times 2=\mim(A)\mim(B)
\]
for 
\[
A=\begin{pmatrix}
1&0&1&0&0&0\\
1&1&0&0&0&0\\
0&1&1&0&0&0\\
1&1&1&1&0&1\\
1&1&1&1&1&0\\
1&1&1&0&1&1\\
\end{pmatrix}
\quad \text{ and }\quad
B=\begin{pmatrix}
1&1&0&1&1&1\\
0&1&1&1&1&1\\
1&0&1&1&1&1\\
0&0&0&1&1&0\\
0&0&0&0&1&1\\
0&0&0&1&0&1\\
\end{pmatrix}.
\]
The property of being a half-induced matching can also be weakened to $M[i,i]=1$ and $M[i,j]+M[j,i]\leq 1$ for all $i\neq j$ (`sub-antisymmetric'). Let $g(M)$ denote the largest submatrix with this property. Then $g(M)\geq \hg(M)$ and $\hg(M)\geq \lfloor \log_2(g(M))\rfloor$ by greedily selecting the row $r$ with the most zeros, throwing away the columns that have a 1 entry in row $r$. Each time we throw out at most halve of the remaining columns, so this can continue for at least  $\lfloor \log_2(g(M))\rfloor$ steps. This shows these two parameters are functionally equivalent.

\subsection{Comparing $\hg$ and support rank}
\label{subsec:hgvssuprank}
Our algorithm in \cref{lem:RepOne} is based on the size of the largest half-induced matching. We next show that $\hg$ and the support rank are not functionally equivalent. In particular, this shows our new approach gives a better guarantee when $k$ is small (when $H$ may be chosen arbitrarily from the infinite family of matrices below).
\begin{theorem}
\label{thm:separation_hg_minrank}
    For each field $\mathbb{F}$ and integer $t\geq 3$, for all $n$ sufficiently large, there exists a symmetric matrix $A\in \{0,1\}^{n\times n}$ with ones on the diagonal, such that
    \begin{itemize}
        \item $\hg(A)\leq 2t$, and
        \item any matrix $B\in \mathbb{F}^{n\times n}$ with the same support as $A$ has rank at least $n^{1-2/t}$.  
    \end{itemize}
\end{theorem}
Rather than proving this result ourselves, we will conclude it from results for a related notion.
The \textit{minrank} of a graph $G$ on the set of vertices $[n]$ over a field $\mathbb{F}$ is the minimum possible rank of a matrix $M\in \mathbb{F}^{n\times n}$ with nonzero diagonal entries such that $M_{i,j}=0$ whenever $i$ and $j$ are distinct nonadjacent vertices of $G$. This gives an upper bound on the Shannon capacity of the graph, but can also be used to give lower bounds on the smallest dimension $d$ for which the graph admits certain geometric representations in $\mathbb{R}^d$, e.g. orthogonal representations, unit distance graphs or touching spheres (see 
\cite{GolovnevRegevWeinstein18,Haviv19,AlonBallaGishbolinerMondMousset20}).


If we denote by $A\in \{0,1\}^{n\times n}$ the adjacency matrix of $G$ and write  
\[
\mathcal{M}(A,\mathbb{F})=\{M \in \mathbb{F}^{n\times n} \ | \ A_{i,j}=0
\implies M_{i,j}=0 \text{ for }i\neq j\text{ and }M_{i,i}\neq 0 \text{ for all }i \}
\]
then the minrank of $G$ is $\min\{\text{rank}(M) \ | \ M\in \mathcal{M}(A,\mathbb{F})\}$. With
\[
\mathcal{M}'(A,\mathbb{F})=\{M \in \mathbb{F}^{n\times n} \ | \ M_{i,j}=0\iff A_{i,j}=0\},
\]
the support rank of $A$ is $\min\{\text{rank}(M) \ | \ M\in \mathcal{M}'(A)\}$.
We see that the set $\mathcal{M}(A,\mathbb{F})$ remains the same when all the diagonal entries of $A$ are changed to ones. 
When $A$ has ones on the diagonal, then $\mathcal{M}'(A,\mathbb{F})\subseteq \mathcal{M}(A,\mathbb{F})$ and so the support rank is lower bounded by the minrank of $A$.

For a graph $H$ with $h \geq 3$
vertices, let $m_2(H)$ denote the maximum value of $\frac{|E(H')|-1}{|V(H')|-2}$ over all subgraphs $H'$ of $H$ on at least 3 vertices. We will use the following result of Alon, Balla, Gishboliner, Mond and Mousset \cite{AlonBallaGishbolinerMondMousset20}.
\begin{theorem}[{\cite[Theorem 5.3]{AlonBallaGishbolinerMondMousset20}}]
\label{thm:Alon_minrank}
Let $H$ be a graph with $h \geq  3$ vertices.
Then there is a constant $c = c(H) > 0$
such that for every (finite or infinite) field $\mathbb{F}$ and every integer $n$ there is a graph $G$ on $n$ vertices
whose complement contains no copy of $H$, so that
$\minrank_{\mathbb{F}}(G) \geq c
n^{1-1/m_2(H)}/\log n$.
\end{theorem}
Let $t\geq 3$.
For $K_{t,t}$ the complete bipartite graph with sides of size $t$, 
\[
m_2(K_{t,t})\geq \frac{t^2-1}{2(t-1)}=\frac{t+1}2.
\]
By Theorem \ref{thm:Alon_minrank} with $H=K_{t,t}$, we find that there is a graph $G$ on $n$ vertices of minrank at least $c(K_{t,t})n^{1-1/m_2(K_{t,t})}/\log n$ avoiding $K_{t,t}$ in its complement. When $n$ is chosen to be sufficiently large, 
\[
c(K_{t,t})n^{1-1/m_2(K_{t,t})}/\log n>n^{1-2/t}.
\]
Let $A$ be the adjacency matrix of $G$ where ones are placed on the diagonal. By our previous discussion, the support rank of $A$ is at least $n^{1-2/t}$.

We claim that $\hg(A)<2t$. Indeed, if $A$ contained a half-induced matching of size $2t$, then there are $t$ rows $r_1,\dots,r_t$ and $t$ columns $c_1,\dots,c_t$ that form an all-zero block in $A$. In particular, this contains no diagonal entries so $r_i\neq c_j$ for all $i,j\in [t]$. But then $\overline{G}[\{r_1,\dots,r_t,c_1,\dots,c_t\}]$ has a complete bipartite graph $K_{t,t}$ as subgraph, a contradiction. This proves Theorem \ref{thm:separation_hg_minrank}.

\subsection{Comparing $\mimsup$ and Shannon capacity}
\label{sec:shannon}
In this section, we shortly discuss the relation between our new parameter and the Shannon capacity. 
Recall that the Shannon capacity is defined as
\[
\Theta(G)=\limsup_{k\to\infty}\alpha(G^{\boxtimes k}),
\]
where $\alpha(G)$ is the independence number of $G$ and $G^{\boxtimes k}$ denotes the result of taking the strong graph product of $k$ copies of $G$ (in the strong product two vertices are adjacent if on each coordinate they contain either a pair neighbors or the same vertex, see~\cite{hammack2011handbook}). 

Let $G$ be a graph with line graph $L(G)$. Denote $(L(G))^2$ for the square of the line graph: the vertex set is $E(G)$ where $e,e'\in E(G)$ are adjacent if $e\cap e''$ and $e'\cap e''$ are both non-empty for some $e''\in E(G)$ (possibly $e'' \in \{e,e'\}$). Then there is a one-to-one correspondence between independent sets in $L(G)^2$ and induced matchings of $G$: two edges $e\neq e'$ can be together in an induced matching if and only if no edge $e''$ intersects both $e$ and $e'$. 

However, $L(G\otimes H)^2$ is not isomorphic to $L(G)^2\boxtimes L(H)^2$ and in fact $\mimsup(G)$ and $\Theta(L(G)^2)$ are not functionally related as the following example shows.
Let $H$ be a half-graph on vertices $v_1,\dots,v_r$ on one side and $w_1,\dots,w_r$ on the other, where the edges are $v_iw_j$ for $i\leq j$. Consider two edges $e=v_iw_j$ and $f=v_sw_r$ (so $i\leq j$ and $s\leq r$). If $i=s$ or $j=r$, then $ef$ is an edge in $L(H)$. Otherwise, if $i<s\leq r$ then there is an edge $v_iw_r$ in $H$, and thus $ef$ is an edge in $L(H)^2$. Finally, if $s<i\leq j$, then $v_sw_j$ is an edge in $H$ and therefore $ef$ is an edge in $L(H)^2$. So $(L(H))^2$ is the complete graph on $\frac{r(r+1)}{2}$ vertices. Any strong product power $((L(G))^2)^{\boxtimes k}$ is also a complete graph, so the Shannon capacity of $(L(H))^2$ is $1$, while $\mimsup(H)\geq \hg(H)=r$. 

On the other hand, we do find the following relationship.
\begin{lemma}
    For any bipartite graph $G$ it holds that $\Theta(L(G)^2)\leq \mimsup(G)$.
\end{lemma}
\begin{proof}
Let $G$ be a bipartite graph with bipartition classes $U,V$.
Let $A$ be the bi-adjacency matrix of $G$. Let $k\geq 1$ be an integer. Let $G_k$ be the bipartite graph corresponding to bi-adjacency matrix $A^{\otimes k}$ on vertex sets $U^k$ and $V^k$. (In our notation $G_k=G^{\otimes k}$.) It suffices to show that
\[
\alpha((L(G)^2)^{\boxtimes k})\leq \alpha(L(G_k)^2)=\mim(G_k)=\mim(A^{\otimes k}),
\]
where the first equality has been shown at the start of this section and the second is by definition. We prove the inequality above by showing $G_1=L(G_k)^2$ is isomorphic to a subgraph of $G_2=(L(G)^2)^{\boxtimes k}$ -- note that both graphs have the same number of vertices, so in fact we will show that $G_2$ is isomorphic to a graph obtained by adding some (possibly zero) edges to $G_1$.

The vertices of $L(G_k)^2$ are of the form $\{u,v\}$ with $u=(u_1,\dots,u_k)\in U^k$ and $v=(v_1,\dots,v_k)\in V^k$ and $e_i=u_iv_i\in E$ for all $i\in [k]$. We let $f(\{u,v\})=(e_1,\dots,e_k)\in V((L(G)^2)^{\boxtimes k})$. Then $f:V(G_1)\to V(G_2)$ is an injective function. It remains to show that $f(e)f(e')\in E(G_2)$ when $ee'\in E(G_1)$. Let $ee'\in E(G_1)$ be given with $e=\{u,v\}$ and $e'=\{u',v'\}$ for $u,u'\in U^k$ and $v,v'\in V^k$.
Since $ee'\in E(L(G_k)^2)$, either $\{u,v'\}$ or $\{u',v\}$ must be an element of $E(G_k)$. By symmetry, we may assume $e''=\{u,v'\}\in E(G_k)$. But by definition, that means that for all $i\in [k]$, $e_i'':=u_iv_i'\in E$ is incident with both  $e_i:=u_iv_i\in E$ and $e_i':=u_i'v_i'\in E$. So $(e_1,\dots,e_k)$ is adjacent to $(e_1',\dots,e_k')$ in $G_2$. The claim now follows as $f(e)=(e_1,\dots,e_k)$ and $f(e')=(e_1',\dots,e_k')$.
\end{proof}

\subsection{Support rank, covering by bicliques, and Prague dimension}

Let $H$ be a graph on $n$ vertices.
The \emph{Prague dimension} of a graph $H$ (sometimes also called \emph{product dimension} or just \emph{dimension}) introduced by Ne\v{s}et\v{r}il, Pultr, and R\"odl~\cite{10.1007/3-540-08442-8_119,NESETRIL197849} is the least integer $p$ such that there exist integers $n_1,\ldots,n_p$ for which $H$ is an induced subgraph of $K_{n_1}\times \ldots \times K_{n_p}$, where $\times$ denotes the direct product (see \cref{sec:algofinal}) and $K_{n_i}$ is the complete graph on $n_i$ vertices.
Note that without loss of generality we may assume that $n_1=n_2=\ldots=n_p=n$.
We denote the Prague dimension of $H$ by $\dim(H)$. Note that for two vertices $u=(u_1,\ldots,u_p), v=(v_1,\ldots,v_p)$ of $K_{n}\times \ldots \times K_{n}$, they are adjacent if and only if $u_i\neq v_i$ for every $i\in [p]$. Therefore equivalently we can say that $\dim(H)\leq p$ if we can encode each vertex of $H$ as a sequence of length $p$ so that the vertices are adjacent if and only if their corresponding sequences differ on every coordinate.

\begin{theorem}\label{thm:dim-bicliques}
Let $H$ be a bipartite graph.
\begin{enumerate}[(1)]
\item There exists a family $\cB$ that $\dim(H)$-covers $H^c$.
\item If $H^c$ can be $r$-covered, then $\dim(H)\leq r^2+2$.
\end{enumerate}
\end{theorem}

\begin{proof}[Proof of (1).]
   Let $n=|V(H)|$, $r=\dim(H)$, and let $\mu: V(H) \to [n]^r$ be a mapping whose image induces a copy of $H$ in $r$-fold direct product of $K_n$ with vertex set $[n]$; it exists by the definition of $\dim(H)$.
   For $i\in [r], j\in [n]$, let $B_{i,j}$ be the subgraph of $H^c$ induced by these vertices $v\in V(H)$ for which the projection of $\mu(v)$ to the $i$-th coordinate is equal to $j$.
   We define $\cB = \{B_{i,j} ~|~ i \in [r], j \in [n] \}$.

   Observe that all vertices of $B_{i,j}$ are pairwise non-adjacent in $H$ since for all tuples $\mu(v)$ their $i$-th coordinate is the same.
   Therefore $B_{i,j}$ is an induced biclique of $H^c$.

   Now let us verify that $\cB$ $r$-covers $H^c$. First consider a vertex $v\in V(H^c)=V(H)$ and let $\mu(v)=(v_1,\ldots,v_p)$.
   The vertex $v$ is contained is subgraph $B_{i,v_i}$ for $i\in [r]$, so $v$ is in $r$ bicliques.
   So now consider any edge $uv$ of $H^c$, and let $\mu(u)=(u_1,\ldots,u_r)$ and $\mu(v)=(v_1,\ldots,v_r)$.
   By the definition of $H^c$, we have that $uv\notin E(H)$ and therefore there exists $i\in [r]$ such that $u_i=v_i$. Thus both $u,v$ are contained in the biclique $B_{i,u_i}$.
   This completes the proof.
\end{proof}
\begin{proof}[Proof of (2).]
Let the bipartition classes of $H$ be $X,Y$. 
Let $\cB$ be a family of bicliques that $r$-covers $H^c$ and let $B_1,\ldots,B_s$ be the fixed arbitrary ordering of $\cB$. 
Define $I=[r]^2\cup \{(0,1),(1,0)\}$, which will be our set of indices. Note that $|I|=r^2+2$.
Let $q = \max(|V(H)|,s+2)$.
We aim to define a mapping $\mu: V(H) \to [q]^{|I|}$ which will define an induced copy of $H$ in $|I|$-fold product of $K_{q}$ with vertex set $[q]$ (here it will be convenient not to assume that the cliques have $|V(H)|$ vertices).

Fix a pair $(i,j) \in [r]^2$ and consider $u \in X$ (resp. $v \in Y$).
Let $u$ (resp. $v$) be covered by $r' \leq r$ bicliques in $\cB$.
We define
\[
\mu_{i,j}(u) = \begin{cases}
    \ell & \text{ if $i \leq r'$ and $B_\ell$ is the}\\
    & \text{$i$-th biclique covering $u$}\\
    s+1 & \text{ if } i > r',
\end{cases}
\quad
\mu_{i,j}(v) = \begin{cases}
    \ell & \text{ if $j \leq r'$ and $B_\ell$ is the}\\
    & \text{$j$-th biclique covering $v$}\\
    s+2 & \text{ if } i > r'.
\end{cases}
\]
Moreover, we define $\mu_{1,0}: V(H) \to [q]$ so that each vertex receives distinct value (this is possible as $q \geq |V(H)|$),
and $\mu_{0,1}: V(H) \to [q]$ that maps all vertices from $X$ to $1$ and all vertices from $Y$ to $2$.

Now, $\mu : V(H) \to [q]^{I}$ is defined in a way that the projection of $\mu$ on the coordinate $(i,j) \in I$ is precisely $\mu_{i,j}$.
This completes the definition of $\mu$.

Let us verify that the image $\mu(V(H))$ induces a copy of $H$.
By the definition of $\mu_{1,0}$, the mapping is injective and by the definition of $\mu_{0,1}$ for $u,v$ from the same bipartition class we have that $\mu(u)$ and $\mu(v)$ are non-adjacent in $K_{q} \times \ldots \times K_{q}$ (we remark that $\mu_{1,0}$ is not needed if all vertices have pairwise distinct neighborhoods). 

Consider $uv\in E(H)$ with $u\in X$, $v\in Y$.
This means that $uv\notin E(H^c)$ and thus there is no biclique that contains both $u,v$.
Therefore $\mu(u)$ and $\mu(v)$ differ on every coordinate and hence are adjacent.

So now consider $uv\notin E(H)$ with $u\in X$, $v\in Y$. Since $uv\in E(H^c)$, there must be at least one biclique $B_\ell\in \cB$ that contains $uv$.
Therefore for some pair $(i,j)$ we defined $\mu_{i,j}(u)=\mu_{i,j}(v)=\ell$, so $\mu(u)$ is non-adjacent to $\mu(v)$. This completes the proof.
\end{proof}

We conclude this section with the following bound.

\begin{corollary}\label{cor:suprank-dim}
The following bounds hold.
\begin{enumerate}
    \item If $H$ is not bipartite, then the support rank of the adjacency matrix of $H$ is at most {$(\dim(H^*)+1)^{\dim(H^*)}$}.
    \item If $H$ is bipartite, then the support rank of the bi-adjacency matrix of $H$ is at most\\ {$(\dim(H)+1)^{\dim(H)}$}.
\end{enumerate}
\end{corollary}
\begin{proof}
    The first statement is an immediate corollary of \cref{lem:suprank-dim} and \cref{thm:dim-bicliques}~(1).
    We point out that \cref{lem:suprank-dim} holds, if we consider covering $H^c$ instead of $(H^*)^c$ and the bi-adjacency matrix of $H$ instead the adjacency matrix. Therefore, the second statement follows as well.
\end{proof}

\section{Conclusion}
\label{sec:conc}

We conclude the paper by discussing graphs $H$ that are potentially missed by Theorem~\ref{thm:mainalgohomo} and Theorem~\ref{thm:lower}, and pointing out some directions for further research.

\paragraph{Non-projective graphs.}
The combination of Corollary~\ref{cor:algohomoprojective} and  Theorem \ref{thm:lower} gives strong evidence that if $H$ is projective, then indeed $\mimsup(H)$ is the ``right'' parameter to study in the context of~\cref{q1}. So let us consider the case that $H$ is  not projective.

There are two questions to be asked: (i) what graphs are non-projective non-bipartite cores, and (ii) whether we can show some interesting lower bound for \homo{H} in such a case.

Concerning question (i), it turns out that the only way of constructing non-projective graphs that we know is by using direct products.
Indeed, for $H = H_1 \times H_2$, a function that maps each $((x_1,x_2),(y_1,y_2)) \in V(H^2)$ to $(x_1,y_2)$, is a homomorphism from $H^2$ to $H$, but not a projection. Actually, it is conjectured that all non-projective graphs are constructed in such a way~\cite{larose2001strongly,DBLP:journals/siamcomp/OkrasaR21}, i.e., every non-projective non-bipartite core is not prime.
Recall that direct products are handled by Theorem~\ref{thm:mainalgohomo}, and the parameter considered there is actually the maximum value of $\mimsup(H')$, where $H'$ iterates over factors of the prime decomposition of $H$. 

Concerning question (ii), Larose~\cite{Larose2002FamiliesOS} defined a subclass of projective graphs called \emph{strongly projective graphs} (we do not include the definition as it is quite technical and not so relevant to our work). It turns out that we do not know any graphs that are projective but not strongly projective and it is conjectured that these two classes actually coincide~\cite{Larose2002FamiliesOS,DBLP:journals/siamcomp/OkrasaR21}. This conjecture, combined with the one mentioned in the previous paragraph and \cref{thm:unique-fact}, yields that every non-bipartite core graph admits a (unique) prime factorization into strongly projective graphs.

Using the methods introduced in~\cite{DBLP:journals/siamcomp/OkrasaR21}, it is possible to lift Theorem~\ref{thm:lower} to such a setting, and the constant appearing in the lower bound becomes again $\mimsup(H')$ defined as above. As the main focus of the paper is how asymptotic rank parameters can be used to upper-bound the complexity of an algorithm for finding graph homomorphisms, we skip discussing this generalization in detail and refer the reader to~\cite{DBLP:journals/siamcomp/OkrasaR21}.

\paragraph{Directions for further research.}
An obvious open problem is to fully resolve~\cref{q1}.
As discussed, to achieve this goal we only lack a fast algorithm that computes representative sets for partial solutions to \homo{H}. We have shown in Lemma~\ref{lem:RepOne} that such sets can be found non-trivially fast and hope it can be further improved to settle~\cref{q1}. 
A far more ambitious (and probably currently out of reach) goal is to provide a (more) fine-grained version of the Courcelle's theorem for deciding any graph property definable in the monadic second-order logic. While being homomorphic to a given graph $H$ is of course only a very special sort of such a property, we find our progress on~\cref{q1} encouraging in this respect and hope that eventually similar connections between the complexity of more general computational problems and asymptotic rank parameters can be made as well.
In particular, we believe that $\mimsup$ (or a similar parameter that tracks the asymptotic behavior under appropriate products) is likely to determine the limit of dynamic programming approaches in other settings as well, especially those determined by various graph width parameters.

\medskip
Another suggested direction of research is purely combinatorial/algebraic: we expect that $\mimsup$ is an interesting parameter for further study in its own right. Similar asymptotic parameters have been defined for tensors (such as asymptotic tensor rank, which can be used to define the matrix multiplication constant $\omega$) but the interesting asymptotic aspects disappear for the special case of matrices. We give some suggestions below.
\begin{itemize}
    \item What type of values can $\mimsup(H)$ take given a matrix $H$? Can it take non-integer values? Similar questions have recently been investigated for asymptotic tensor parameters, see e.g.~\cite{blatter22,brietea23}.
    \item What is the value of $\mimsup$ for a $n\times n$ random matrix, where each entry of the matrix gets sampled independently to be $1$ with probability $p$ and to be $0$ with probability $1-p$?
    \item  We showed that $\hg$ and the support rank are not functionally equivalent. Is $\mimsup$ functionally equivalent to either $\hg$ or the support rank? 
\end{itemize}
The second question may shed some light on the third. 

Our own progress on the relation between $\hg$ and $\mimsup$ lifted basic arguments from (multi-color) Ramsey numbers.
Let $R(t;\ell)$ denote the smallest integer $r$ for which every $\ell$-coloring of $K_r$, the complete graph on $r$ vertices, has a monochromatic $K_t$. Ignoring smaller order terms, for $\ell\geq 3$, the best upper bound $\ell^{\ell t}$ comes from the `neighbourhood chasing argument' similar to the one we applied in \cref{lem:RepOne}. On the lower bound side it is difficult to get beyond bounds of the form $c^{\ell t}$, for $c$ constant (see \cite{ConlonFerber} for further discussion). It remains open whether $g(t)=\limsup_{\ell \to \infty} R(t;\ell)^{1/\ell}$ is finite for all constant $t$. In particular, the state-of-the-art from this line of work may give some improvement on our separation but is unlikely to (directly) yield an answer to the problem of whether $\hg$ and $\mimsup$ are functionally equivalent or not.  
\paragraph{Acknowledgement.} The authors are grateful to Koblich for enlightening discussions about communication complexity.

\bibliographystyle{plain}
\bibliography{main}
\end{document}